\documentclass[a4paper,superscriptaddress,11pt, shorttitle=template]{compositionalityarticle}
\pdfoutput=1

\usepackage{mathtools} 
\usepackage{amssymb} 
\usepackage{bbold} 
\usepackage{amsthm} 

\usepackage{microtype} 
\usepackage{ragged2e} 

\usepackage[round,authoryear,sort&compress]{natbib} 

\usepackage{tikz} 
\usepackage{circuitikz} 
\usetikzlibrary{
	arrows,
	shapes,
	decorations,
	intersections,
	backgrounds,
	positioning,
	circuits.ee.IEC
	}

		\newcounter{theorem_c} 

		\theoremstyle{plain} 
		\newtheorem{theorem}[theorem_c]{Theorem}

		\newtheoremstyle{exampstyle}
		  {2mm} 
		  {2mm} 
		  {\itshape} 
		  {} 
		  {\bfseries} 
		  {.} 
		  {.5em} 
		  {} 

		\theoremstyle{exampstyle}
		\newtheorem{definition}[theorem_c]{Definition}

	\newcommand{\inlineQuote}[1]{\textquotedblleft #1\textquotedblright} 

	\newcommand{\reals}{\mathbb{R}} 
	
	\newcommand{\singletonSet}{\mathbb{1}}

		\newcommand{\ket}[1]{\left\vert #1 \right\rangle} 
		\newcommand{\bra}[1]{\langle #1 \vert} 
		
		\newcommand{\SpaceH}{\mathcal{H}} 
		\newcommand{\SpaceG}{\mathcal{G}}

		\newcommand{\RMatCategory}[1]{#1\operatorname{-Mat}} 

		\newcommand{\CategoryC}{\mathcal{C}}

		\newcommand{\obj}[1]{\operatorname{obj} \, #1} 

		\newcommand{\classicalSubcategory}[1]{#1_{K}} 

	\newcommand{\Xcolour}{Red}
	
	\newcommand{\Zcolour}{YellowGreen}

	\newcommand{\Xaltcolour}{Purple}
	
	\newcommand{\Zaltcolour}{Cyan}
	
	\newcommand{\Xbwcolour}{black!80}

	\newcommand{\Zbwcolour}{white}

	\newcommand{\Ybwcolour}{black!15}

	\newcommand{\Wbwcolour}{black!55}

	\newcommand{\trace}[1]{\hbox{\begin{tikzpicture} [scale=1.2,transform shape,rotate=-90] 

\def\deltax{0.3} 
\def\deltay{0.5} 

\path[use as bounding box] (-\deltax,-0.5*\deltay) rectangle (\deltax,0.65*\deltay);

\node (mult) at (0,0.3*\deltay) [upground,scale=0.5] {};
\node (mult_label_in) at (0,-0.7*\deltay) {};
\draw[-] (mult_label_in) to (mult);

\end{tikzpicture}}\!_{#1}} 

	\tikzset{
	  rectangle with rounded corners north west/.initial=4pt,
	  rectangle with rounded corners south west/.initial=4pt,
	  rectangle with rounded corners north east/.initial=4pt,
	  rectangle with rounded corners south east/.initial=4pt,
	}
	\makeatletter
	\pgfdeclareshape{rectangle with rounded corners}{
	  \inheritsavedanchors[from=rectangle] 
	  \inheritanchorborder[from=rectangle]
	  \inheritanchor[from=rectangle]{center}
	  \inheritanchor[from=rectangle]{north}
	  \inheritanchor[from=rectangle]{south}
	  \inheritanchor[from=rectangle]{west}
	  \inheritanchor[from=rectangle]{east}
	  \inheritanchor[from=rectangle]{north east}
	  \inheritanchor[from=rectangle]{south east}
	  \inheritanchor[from=rectangle]{north west}
	  \inheritanchor[from=rectangle]{south west}
	  \backgroundpath{
	    \southwest \pgf@xa=\pgf@x \pgf@ya=\pgf@y
	    \northeast \pgf@xb=\pgf@x \pgf@yb=\pgf@y
	    \pgfkeysgetvalue{/tikz/rectangle with rounded corners north west}{\pgf@rectc}
	    \pgfsetcornersarced{\pgfpoint{\pgf@rectc}{\pgf@rectc}}
	    \pgfpathmoveto{\pgfpoint{\pgf@xa}{\pgf@ya}}
	    \pgfpathlineto{\pgfpoint{\pgf@xa}{\pgf@yb}}
	    \pgfkeysgetvalue{/tikz/rectangle with rounded corners north east}{\pgf@rectc}
	    \pgfsetcornersarced{\pgfpoint{\pgf@rectc}{\pgf@rectc}}
	    \pgfpathlineto{\pgfpoint{\pgf@xb}{\pgf@yb}}
	    \pgfkeysgetvalue{/tikz/rectangle with rounded corners south east}{\pgf@rectc}
	    \pgfsetcornersarced{\pgfpoint{\pgf@rectc}{\pgf@rectc}}
	    \pgfpathlineto{\pgfpoint{\pgf@xb}{\pgf@ya}}
	    \pgfkeysgetvalue{/tikz/rectangle with rounded corners south west}{\pgf@rectc}
	    \pgfsetcornersarced{\pgfpoint{\pgf@rectc}{\pgf@rectc}}
	    \pgfpathclose
	 }
	}
	\makeatother

	\tikzset{->-/.style={decoration={markings,mark=at position #1 with {\arrow{>}}},postaction={decorate}}}
	\tikzset{-<-/.style={decoration={markings,mark=at position #1 with {\arrow{<}}},postaction={decorate}}}

	\tikzstyle{every picture}=[baseline=-0.25em,scale=0.5]
	\pgfdeclarelayer{edgelayer}
	\pgfdeclarelayer{nodelayer}
	\pgfsetlayers{edgelayer,nodelayer,main}

	\tikzstyle{box} = [draw,shape=rectangle,inner sep=2pt,minimum height=6mm,minimum width=6mm,fill=white] 
	\tikzstyle{boxlarge} = [draw,shape=rectangle,inner sep=2pt,minimum height=1.5cm,minimum width=8mm,fill=white] 
	\tikzstyle{boxLarge} = [draw,shape=rectangle,inner sep=2pt,minimum height=2cm,minimum width=10mm,fill=white] 
	\tikzstyle{boxLargerThin} = [draw,shape=rectangle,inner sep=2pt,minimum height=3cm,minimum width=5mm,fill=white] 
	\tikzstyle{boxsmall} = [draw,shape=rectangle,inner sep=2pt,minimum height=3mm,minimum width=3mm,fill=white] 
	\tikzstyle{dot} = [inner sep=0mm,minimum width=3mm,minimum height=3mm,draw,shape=circle,text depth=-0.1mm]
	\tikzstyle{Zbwdot} = [dot, fill=\Zbwcolour]
	\tikzstyle{Xbwdot} = [dot, fill=\Xbwcolour]
	\tikzstyle{Ybwdot} = [dot, fill=\Ybwcolour]
	\tikzstyle{Wbwdot} = [dot, fill=\Wbwcolour]
	\tikzstyle{antipode} = [boxsmall] 

	\tikzstyle{state} = [draw, rectangle with rounded corners,
	  rectangle with rounded corners north west=8pt,
	  rectangle with rounded corners south west=8pt,
	  rectangle with rounded corners north east=0pt,
	  rectangle with rounded corners south east=0pt,
	,inner sep=2pt,minimum height=6mm,minimum width=6mm,fill=white]
	\tikzstyle{statelarge} = [draw, rectangle with rounded corners,
	  rectangle with rounded corners north west=8pt,
	  rectangle with rounded corners south west=8pt,
	  rectangle with rounded corners north east=0pt,
	  rectangle with rounded corners south east=0pt,
	,inner sep=2pt,minimum height=1.5cm,minimum width=8mm,fill=white]
	\tikzstyle{stateLarge} = [draw, rectangle with rounded corners,
	  rectangle with rounded corners north west=8pt,
	  rectangle with rounded corners south west=8pt,
	  rectangle with rounded corners north east=0pt,
	  rectangle with rounded corners south east=0pt,
	,inner sep=2pt,minimum height=2cm,minimum width=8mm,fill=white]
	\tikzstyle{stateLargerThin} = [draw, rectangle with rounded corners,
	  rectangle with rounded corners north west=8pt,
	  rectangle with rounded corners south west=8pt,
	  rectangle with rounded corners north east=0pt,
	  rectangle with rounded corners south east=0pt,
	,inner sep=2pt,minimum height=3cm,minimum width=5mm,fill=white]
	\tikzstyle{effect} = [draw, rectangle with rounded corners,
	  rectangle with rounded corners north west=0pt,
	  rectangle with rounded corners south west=0pt,
	  rectangle with rounded corners north east=8pt,
	  rectangle with rounded corners south east=8pt,
	,inner sep=2pt,minimum height=6mm,minimum width=6mm,fill=white]
	\tikzstyle{effectLargerThin} = [draw, rectangle with rounded corners,
	  rectangle with rounded corners north west=0pt,
	  rectangle with rounded corners south west=0pt,
	  rectangle with rounded corners north east=8pt,
	  rectangle with rounded corners south east=8pt,
	,inner sep=2pt,minimum height=3cm,minimum width=5mm,fill=white]
	\tikzstyle{scalar}=[diamond,draw,inner sep=1pt,font=\small,fill=white]

	\tikzstyle{cdnode}=[fill=white]
	\tikzstyle{labelnode}=[fill=white]
	\tikzstyle{tightlabelnode}=[fill=white,inner sep = 0.1mm]
	\tikzstyle{none}=[inner sep=0pt]
	\tikzstyle{whiteline}=[-, line width=4pt, draw=white]

	\tikzstyle{trace}=[circuit ee IEC,thick,ground,scale=2.5]
	\tikzstyle{cotrace}=[circuit ee IEC,thick,ground,rotate=180,scale=2.5]
	\tikzstyle{upground}=[circuit ee IEC,thick,ground,rotate=90,scale=2.5]
	\tikzstyle{downground}=[circuit ee IEC,thick,ground,rotate=-90,scale=2.5]

	\tikzstyle{doubled} = [line width=1.8pt] 
	
	\tikzstyle{empty diagram}=[draw=gray!40!white,dashed,shape=rectangle,minimum width=1cm,minimum height=1cm]

\begin{document}

\title{Purification and time-reversal deny entanglement in LOCC-distinguishable orthonormal bases}
\author{Stefano Gogioso}
\orcid{0000-0001-7879-8145}
\email{stefano.gogioso@cs.ox.ac.uk}
\affiliation{Quantum Group, Oxford University}
\author{Subhayan Roy Moulik}
\email{subhayan.roy.moulik@cs.ox.ac.uk}
\affiliation{Quantum Group, Oxford University}
\date{\today}

\newcommand{\commentSRM}[1]{\textcolor{blue}{SRM says: #1}}
\newcommand{\commentSG}[1]{\textcolor{red}{SG says: #1}}

\newcommand{\cf}{cf. }
\newcommand{\eg}{e.g. }


\maketitle

\begin{abstract}
\noindent We give a simple proof, based on time-reversibility and purity, that a complete orthonormal family of pure states which can be perfectly distinguished by LOCC cannot contain any entangled state. Our results are really about the shape of certain states and processes, and are valid in arbitrary categorical probabilistic theories with time-reversal. From the point of view of the resource theory of entanglement, our results can be interpreted to say that free processes can distinguish between the states in a complete orthonormal family only when the states themselves are all free.
\end{abstract}

\section{Introduction}

Quantum theory allows perfect deterministic distinguishability of orthogonal states: for example, there is a POVM (positive operator valued measure) which, given a state from the orthogonal set $\{\ket{00}, \ket{11}, \ket{01}+\ket{10},\ket{01}-\ket{10}\}$ as its input, returns a classical output in the set $\{0,1,+,-\}$ uniquely identifying the state. When the system under consideration is multipartite (e.g. the bipartite system above), the POVM will in general contain non-separable effects: it is therefore interesting to investigate which limitations are imposed on distinguishability tasks by the requirement that quantum operations be local to the parties involved, but allowing arbitrary amounts of classical communication between the parties, which can inform the classical control of the local quantum operations (i.e. we wish to consider LOCC scenarios). 

Investigation on the power of LOCC scenarios was initiated by Peres and Wooters, leading to the development of their much-celebrated teleportation protocol (\cf \cite{P04}). However, it could be argued that the cornerstone result which truly boosted the field was the proof that not all sets of orthogonal product states can be distinguished by LOCC scenarios (\cf \cite{Betal99}). This is a very counter-intuitive result, in a sense, as it proves that there is an orthonormal basis of 9 bipartite qutrit states which can be \inlineQuote{prepared locally}---as it only consists of product states---but which cannot be distinguished by LOCC.

The issue of LOCC distinguishability has been studied in more general settings throughout the years, with more LOCC-indistinguishable orthonormal families explicitly constructed (\cf \cite{ZTWL16}), and all bipartite qutrit and tripartite qubit families characterised (\cf \cite{FS09,WH02}). It is also known that it is always possible to distinguish between two orthogonal states with certainty by LOCC (\cf \cite{Wetal00}), but that it is generally impossible to distinguish with certainty between three or more states (\cf \cite{Getal01, Getal02}). Further research includes connections with the theory of entanglement witnesses (\cf \cite{C04}), cryptographic applications such as secret sharing and data hiding (\cf \cite{HM03,DHT03, LW13}), as well as a number of alternative approaches to the problem (\eg \cf \cite{DFXY09, CLMO13, Cetal14, CDH14,RoyMoulik2016}).

The question of which sets of orthogonal quantum states can be perfectly distinguished by LOCC protocols remains open \cite{Cetal14}, as does the quest for the understanding of the physical principles playing a role in this problem. In this work, we focus our efforts on the role played by the physical and informational principles of purity and time-reversal: harnessing their power, we provide a simple alternative proof of a key result by \cite{HSSK03}, which we furthermore generalise beyond quantum theory. We show that if it is possible to distinguish between the pure states of an orthonormal basis by LOCC alone then the basis can only contain product states. Our results can also be recast from the point of view of resource theory (\cf \cite{Coecke2016b}), and specifically of the resource theory of entanglement by \cite{Chiribella2015}: within that framework, our result says that free operations cannot distinguish between the states in an orthonormal basis unless the states themselves are all free.

Our proof is carried out in the formalism of string diagrams for categorical quantum mechanics (\cf \cite{Abramsky2004,Abramsky2009,Backens2014,Coecke2009picturalism,Coecke2011,Coecke2015,Coecke2016,Coecke2017,Horsman2011,Horsman2017}), and more precisely within the recently introduced framework of categorical probabilistic theories (CPTs) by \cite{GS17}. As a consequence, our result extends beyond quantum theory, to any probabilistic theory with time-reversal. The use of diagrammatic methods also highlights an important feature of the result which did not stand out in the original proof of \cite{HSSK03}: LOCC distinguishability is really a story about the \inlineQuote{shape} of certain states and processes, and more specifically a story about how processes of \inlineQuote{LOCC shape} fail at distinguishing between states which are not of \inlineQuote{LOCC shape}.

We have chosen to use categorical probabilistic theories, rather than the more established \textit{operational probabilistic theories} (OPTs, \cf \cite{Hardy2001,Hardy2011,Hardy2011informational,Hardy2013,Chiribella2010,Chiribella2011,Chiribella2014,Chiribella2016,DAriano2017}), for two main reasons. Firstly, the categorical framework is process-oriented, rather than probability-oriented: in CPTs, classical control and classical outcomes can be talked about in an explicitly compositional way, and are naturally related by time-reversal; in OPTs, on the other hand, classical systems are introduced externally using indices, in a non-compositional way. Secondly, the categorical framework provides us with additional flexibility when decomposing \inlineQuote{normalised} processes (e.g. CPTP maps and quantum instruments) in terms of \inlineQuote{unnormalised} building blocks (e.g. generic CP maps and families thereof), especially when it comes to dealing with the time-reversal of the building blocks themselves (which might not be individually sub-normalised); conversely, OPTs impose very specific requirements on the building blocks allowed in processes, introducing the need for unnecessary additional checks and restrictions.

\section{Categorical probabilistic theories}

\paragraph{Categorical Probabilistic Theories.} When talking about a \emph{categorical theory}, we mean a symmetric monoidal category which captures physical systems and the compositional structure of processes between them. In the general spirit of categorical quantum mechanics, we avoid restricting our attention to physical processes alone, but instead we allow the presence of a whole spectrum of idealised, abstract processes that provide building blocks for the physical processes themselves, or otherwise help in reasoning about them. When talking about a \emph{categorical probabilistic theory} (CPT), we mean a symmetric monoidal category which includes at least the classical systems amongst its ranks, and which is furthermore compatible with a couple of basic operational features---namely probabilistic structure and marginalisation. We briefly motivate our requirements below.
\begin{enumerate}
	\item Every probabilistic theory has classical probabilistic systems under the hood: because these are themselves physical systems, we model them explicitly. In particular, their interface with other systems (e.g. measurements and preparations) can be talked about in compositional terms.
	\item It makes no sense to talk about a probabilistic theory if the probabilistic structure does not extend from classical systems to arbitrary systems. If this were not the case, one would not necessarily be able to work with scenarios in which multiplexed processes are controlled by a classical random variable, or to condition a process based on a classical output.
	\item Every probabilistic theory with operational aspirations should include a notion of localisation of states and processes. Indeed, the absence of a notion of local state compatible with classical marginalisation renders most protocol specifications meaningless, a fate they share with the notion of no-signalling and with the probabilistic foundations of thermodynamics. Sensible theories without a notion of local state do exist (e.g. Everettian quantum theory), but they present a number of operational challenges, and they are not probabilistic in nature.
\end{enumerate}
We now proceed to formalise these requirements in categorical terms. When talking about \emph{classical theory}, we mean the symmetric monoidal category  $\RMatCategory{\reals^+}$ which has finite sets $X$ as systems, and where processes $X \rightarrow Y$ are the $Y$-by-$X$ matrices with entries in the non-negative reals $\reals^+$. Processes $X \rightarrow Y$ in classical theory form a convex cone, i.e. they are closed under $\reals^+$-linear combinations. Because we have explicit linear structure, we are allowed to write the following resolution of the identity:
\begin{equation}\label{idClassicalResolution}
\begin{tikzpicture}
	\begin{pgfonlayer}{nodelayer}
		\node [style=labelnode] (0) at (-6, 0) {$X$};
		\node [style=labelnode] (1) at (-2, 0) {$X$};
		\node [style=labelnode] (2) at (4, 0) {$X$};
		\node [style=labelnode] (3) at (10, 0) {$X$};
		\node [style=labelnode] (4) at (0, 0) {$=$};
		\node [style=labelnode] (5) at (2, -0.25) {$\sum\limits_{x \in X}$};
		\node [style=effect] (6) at (6.25, 0) {$x$};
		\node [style=state] (7) at (7.75, 0) {$x$};
	\end{pgfonlayer}
	\begin{pgfonlayer}{edgelayer}
		\draw [style=dashed] (0) to (1);
		\draw [style=dashed] (2) to (6);
		\draw [style=dashed] (7) to (3);
	\end{pgfonlayer}
\end{tikzpicture}
\end{equation}
Classical marginalisation gives rise to a family of effects $(\trace{X}:X \rightarrow \singletonSet)_{X \in \obj{\RMatCategory{\reals^+}}}$ which canonically localises states and processes, and are therefore known as the \emph{discarding maps}:
\begin{equation}\label{discardingMapClassicalResolution}
\begin{tikzpicture}
	\begin{pgfonlayer}{nodelayer}
		\node [style=labelnode] (0) at (-5, 0) {$X$};
		\node [style=trace] (1) at (-2, 0) {};
		\node [style=labelnode] (2) at (4, 0) {$X$};
		\node [style=labelnode] (3) at (0, 0) {$=$};
		\node [style=labelnode] (4) at (2, -0.25) {$\sum\limits_{x \in X}$};
		\node [style=effect] (5) at (6.25, 0) {$x$};
	\end{pgfonlayer}
	\begin{pgfonlayer}{edgelayer}
		\draw [style=dashed] (0) to (1);
		\draw [style=dashed] (2) to (5);
	\end{pgfonlayer}
\end{tikzpicture}
\end{equation}
In the string diagrams literature this is often known as an \emph{environment structure} \cite{Coecke2013a,Coecke2016}, because it satisfies the following equations:
\begin{equation}\label{classicalEnvironmentStructure}
\begin{tikzpicture}[scale=0.8]
	\begin{pgfonlayer}{nodelayer}
		\node [style=labelnode] (0) at (-7.5, 0) {$=$};
		\node [style=labelnode, inner sep=0.1 mm] (1) at (-13.5, 0) {$X\otimes Y$};
		\node [style=trace] (2) at (-9.5, 0) {};
		\node [style=labelnode, inner sep=0.1 mm] (3) at (-5.5, 0.75) {$X$};
		\node [style=trace] (4) at (-2.5, 0.75) {};
		\node [style=labelnode, inner sep=0.1 mm] (5) at (-5.5, -0.75) {$Y$};
		\node [style=trace] (6) at (-2.5, -0.75) {};
		\node [style=labelnode, inner sep=0.1 mm] (7) at (2.5, 0) {$\singletonSet$};
		\node [style=trace] (8) at (5.5, 0) {};
		\node [style=labelnode] (9) at (7.5, 0) {$=$};
		\node [style=empty diagram] (9) at (10.5, 0) {};
	\end{pgfonlayer}
	\begin{pgfonlayer}{edgelayer}
		\draw [style=dashed] (1) to (2);
		\draw [style=dashed] (3) to (4);
		\draw [style=dashed] (5) to (6);
		\draw [style=dashed] (7) to (8);
	\end{pgfonlayer}
\end{tikzpicture}
\end{equation}
The empty diagram on the right hand side of the right equation is simply the scalar 1. Any choice of environment structure singles out a sub-category of \emph{normalised} processes, namely those processes $f$ satisfying the following condition:
\begin{equation}\label{normalisedClassical}
\begin{tikzpicture}[scale=0.8]
	\begin{pgfonlayer}{nodelayer}
		\node [style=box] (0) at (-5.5, 0) {$f$};
		\node [style=labelnode, inner sep=0.1 mm] (1) at (-8.5, 0) {};
		\node [style=trace] (2) at (-2, 0) {};
		\node [style=labelnode] (3) at (0, 0) {$=$};
		\node [style=labelnode, inner sep=0.1 mm] (4) at (2, 0) {};
		\node [style=trace] (5) at (6, 0) {};
	\end{pgfonlayer}
	\begin{pgfonlayer}{edgelayer}
		\draw [style=dashed] (1) to (0);
		\draw [style=dashed] (0) to (2);
		\draw [style=dashed] (4) to (5);
	\end{pgfonlayer}
\end{tikzpicture}
\end{equation}
In classical theory, the normalised processes are exactly the stochastic matrices, and in particular the normalised states are the probability distributions.
\begin{definition}[\textbf{Categorical Probabilistic Theories (\cite{GS17})}]\hfill\\
A \emph{categorical probabilistic theory} is a symmetric monoidal category $\CategoryC$ which satisfies the following three requirements:
\begin{enumerate}
	\item There is a chosen full subcategory of $\CategoryC$, denoted by $\classicalSubcategory{\CategoryC}$, together with a chosen $\reals^+$-linear monoidal equivalence between $\classicalSubcategory{\CategoryC}$ and the category $\RMatCategory{\reals^+}$.
	\item The category $\CategoryC$ has $\reals^+$-linear structure compatible with that which the chosen subcategory $\classicalSubcategory{\CategoryC}$ inherits from  $\RMatCategory{\reals^+}$.
	\item The category $\CategoryC$ has a chosen environment structure $(\trace{\SpaceH}:\SpaceH \rightarrow \singletonSet)_{\SpaceH \in \obj{\CategoryC}}$ compatible with that which the chosen subcategory $\classicalSubcategory{\CategoryC}$ inherits from $\RMatCategory{\reals^+}$.
\end{enumerate}
\end{definition}
\noindent In the context of a specific CPT, we refer to $\classicalSubcategory{\CategoryC}$ as the \emph{classical theory}, to its object as the \emph{classical systems} and to its morphisms as the \emph{classical processes}. In particular, $\singletonSet$ is the tensor unit for $\CategoryC$, and the scalars of $\CategoryC$ form the probabilistic semiring $\reals^+$. A generic process in a CPT can involve both classical systems and more general systems:
\begin{equation}\label{process}
	\begin{tikzpicture}
	\begin{pgfonlayer}{nodelayer}
		\node [style=none] (0) at (-2, -0.25) {};
		\node [style=none] (1) at (-0.5, 0.25) {};
		\node [style=none] (2) at (-2, 0.25) {};
		\node [style=box] (3) at (0, 0) {$M$};
		\node [style=none] (4) at (-0.5, -0.25) {};
		\node [style=none] (5) at (2, 0.25) {};
		\node [style=none] (6) at (0.5, 0.25) {};
		\node [style=none] (7) at (2, -0.25) {};
		\node [style=none] (8) at (0.5, -0.25) {};
		\node [style=tightlabelnode,text=gray] (9) at (-3, 1.5) {generic input};
		\node [style=none] (10) at (-2.25, 0.5) {};
		\node [style=tightlabelnode, text=gray] (11) at (3, 1.5) {generic output};
		\node [style=none] (12) at (2.25, 0.5) {};
		\node [style=tightlabelnode, text=gray] (13) at (-3, -1.5) {classical input};
		\node [style=none] (14) at (2.25, -0.5) {};
		\node [style=tightlabelnode, text=gray] (15) at (3, -1.5) {classical output};
		\node [style=none] (16) at (-2.25, -0.5) {};
	\end{pgfonlayer}
	\begin{pgfonlayer}{edgelayer}
		\draw [style=dashed] (0.center) to (4.center);
		\draw [style=-, line width=4 pt, draw=white, in=180, out=0] (2.center) to (1.center);
		\draw [style=-] (2.center) to (1.center);
		\draw [style=-, line width=4 pt, draw=white, in=180, out=0] (6.center) to (5.center);
		\draw [style=-] (6.center) to (5.center);
		\draw [style=dashed] (8.center) to (7.center);
		\draw [style=->, draw=gray] (9) to (10.center);
		\draw [style=->, draw=gray] (11) to (12.center);
		\draw [style=->, draw=gray] (13) to (16.center);
		\draw [style=->, draw=gray] (15) to (14.center);
	\end{pgfonlayer}
\end{tikzpicture}
\end{equation}
In line with the nomenclature adopted for classical systems, we refer to the $\trace{\SpaceH}$ maps involved in the environment structure as the \emph{discarding maps}, and to those processes $M$ satisfying the following equation as \emph{normalised}:
\begin{equation}\label{processNormalised}
\begin{tikzpicture}
	\begin{pgfonlayer}{nodelayer}
		\node [style=none] (0) at (-7.5, -0.25) {};
		\node [style=none] (1) at (-6, 0.25) {};
		\node [style=none] (2) at (-7.5, 0.25) {};
		\node [style=box] (3) at (-5.5, 0) {$M$};
		\node [style=none] (4) at (-6, -0.25) {};
		\node [style=none] (6) at (-5, 0.25) {};
		\node [style=none] (8) at (-5, -0.25) {};
		\node [style=none] (9) at (0, 0) {$=$};
		\node [style=none] (11) at (1, -0.25) {};
		\node [style=none] (12) at (1, 0.25) {};
	\end{pgfonlayer}
	\begin{pgfonlayer}{edgelayer}
		\node [style=trace] (13) at (3.5, -0.25) {};
		\node [style=trace] (5) at (-3.5, 0.25) {};
		\node [style=trace] (7) at (-2.5, -0.25) {};
		\node [style=trace] (10) at (2.5, 0.25) {};
		\draw [style=dashed] (0.center) to (4.center);
		\draw [style=-] (2.center) to (1.center);
		\draw [style=-, in=180, out=0] (6.center) to (5);
		\draw [style=-, in=180, out=0] (12.center) to (10);
		\draw [style=-, line width=4pt, draw=white] (8.center) to (7);
		\draw [style=-, line width=4pt, draw=white] (11.center) to (13);
		\draw [style=dashed, in=180, out=0] (8.center) to (7);
		\draw [style=dashed, in=180, out=0] (11.center) to (13);
	\end{pgfonlayer}
\end{tikzpicture}
\end{equation}

\paragraph{LOCC instruments.} This work is concerned with a very special class of processes which can me captured by the CPT framework, namely that of LOCC instruments. In its most generic form, an \emph{LOCC instrument} is a process taking the following shape:
\begin{equation}\label{LOCCinstrument}
	\resizebox{0.9\textwidth}{!}{\begin{tikzpicture}
	\begin{pgfonlayer}{nodelayer}
		\node [style=none] (0) at (-13, -2.75) {};
		\node [style=none] (1) at (-11, -0.25) {};
		\node [style=box] (2) at (-10.5, -2.5) {$M_{N}$};
		\node [style=none] (3) at (-11, -2.25) {};
		\node [style=none] (4) at (-13, -0.75) {};
		\node [style=none] (5) at (-13, 2) {};
		\node [style=none] (6) at (-11, 2.5) {};
		\node [style=none] (7) at (-11, -2.75) {};
		\node [style=none] (8) at (-16, -0.25) {};
		\node [style=none] (9) at (-16, 2.5) {};
		\node [style=box] (10) at (-10.5, -0.5) {$M_{N-1}$};
		\node [style=none] (11) at (-11, -0.75) {};
		\node [style=none] (12) at (-10.5, 1) {$\vdots$};
		\node [style=box] (13) at (-10.5, 2.25) {$M_{1}$};
		\node [style=none] (14) at (-16, -2.25) {};
		\node [style=none] (15) at (-11, 2) {};
		\node [style=none] (16) at (-14, 2) {};
		\node [style=none] (17) at (-16, 2) {};
		\node [style=none] (18) at (-14, -0.75) {};
		\node [style=none] (19) at (-16, -0.75) {};
		\node [style=none] (20) at (-14, -2.75) {};
		\node [style=none] (21) at (-16, -2.75) {};
		\node [style=none] (22) at (-5.5, -0.25) {};
		\node [style=none] (23) at (-10, -0.25) {};
		\node [style=none] (24) at (-5.5, -2.25) {};
		\node [style=none] (25) at (-10, -2.25) {};
		\node [style=none] (26) at (-5.5, 2.5) {};
		\node [style=none] (27) at (-10, 2.5) {};
		\node [style=none] (28) at (-8.5, 2) {};
		\node [style=none] (29) at (-10, 2) {};
		\node [style=none] (30) at (-8.5, -0.75) {};
		\node [style=none] (31) at (-10, -0.75) {};
		\node [style=none] (32) at (-8.5, -2.75) {};
		\node [style=none] (33) at (-10, -2.75) {};
		\node [style=tightlabelnode, text=gray] (34) at (-10.5, -4.75) {Local instruments};
		\node [style=none] (35) at (-10.5, -3.75) {};
		\node [style=none] (36) at (-7.5, -2.75) {};
		\node [style=none] (37) at (-5.5, -2.75) {};
		\node [style=none] (38) at (-5.5, -0.75) {};
		\node [style=none] (39) at (-7.5, -0.75) {};
		\node [style=none] (40) at (-5.5, 2) {};
		\node [style=none] (41) at (-7.5, 2) {};
		\node [style=none] (42) at (-8, 3.75) {};
		\node [style=tightlabelnode, text=gray] (43) at (-10.5, 5) {Global classical operations};
		\node [style=none] (44) at (-13.5, 3.75) {};
		\node [style=none] (45) at (-3.5, 0) {$=$};
		\node [style=none] (46) at (-1.5, -0.25) {$\sum\limits_{\underline{x'},\underline{x}}$};
		\node [style=none] (47) at (0, -0.25) {$\sum\limits_{\underline{y},\underline{y'}}$};
		\node [style=state] (48) at (10, 1.5) {\small$x_1$};
		\node [style=none] (49) at (5.5, -2.25) {};
		\node [style=none] (50) at (12.5, 2.5) {};
		\node [style=none] (51) at (12.5, -0.25) {};
		\node [style=effect] (52) at (15.75, -1.25) {\small$y_{N\!-\!1}$};
		\node [style=none] (53) at (20, -0.25) {};
		\node [style=none] (54) at (20, 2.5) {};
		\node [style=state] (55) at (10.25, -1.25) {\small$x_{N\!-\!1}$};
		\node [style=box] (56) at (13, -0.5) {$M_{N-1}$};
		\node [style=none] (57) at (13.5, -0.25) {};
		\node [style=none] (58) at (13, 1) {$\vdots$};
		\node [style=none] (59) at (12.5, -2.75) {};
		\node [style=none] (60) at (12.5, -0.75) {};
		\node [style=effect] (61) at (16, -3.25) {\small$y_N$};
		\node [style=none] (62) at (12.5, 2) {};
		\node [style=none] (63) at (5.5, -0.25) {};
		\node [style=none] (64) at (5.5, 2.5) {};
		\node [style=box] (65) at (13, 2.25) {$M_{1}$};
		\node [style=box] (66) at (13, -2.5) {$M_{N}$};
		\node [style=none] (67) at (12.5, -2.25) {};
		\node [style=none] (68) at (13.5, -2.25) {};
		\node [style=none] (69) at (20, -2.25) {};
		\node [style=state] (70) at (10, -3.25) {\small$x_N$};
		\node [style=none] (71) at (13.5, -2.75) {};
		\node [style=none] (72) at (13.5, 2) {};
		\node [style=none] (73) at (13.5, -0.75) {};
		\node [style=effect] (74) at (16, 1.5) {\small$y_1$};
		\node [style=none] (75) at (13.5, 2.5) {};
		\node [style=none] (78) at (2, 0) {$p_{\underline{x'},\underline{x}}$};
		\node [style=none] (79) at (3.5, 0) {$q_{\underline{y},\underline{y'}}$};
		\node [style=effect] (80) at (8.5, 1.5) {\small$x'_1$};
		\node [style=effect] (81) at (8.25, -1.25) {\small$x'_{N\!-\!1}$};
		\node [style=effect] (82) at (8.5, -3.25) {\small$x'_N$};
		\node [style=none] (83) at (5.5, 2) {};
		\node [style=none] (84) at (5.5, -0.75) {};
		\node [style=none] (85) at (5.5, -2.75) {};
		\node [style=none] (86) at (20, -0.75) {};
		\node [style=state] (87) at (17.5, -3.25) {\small$y'_N$};
		\node [style=none] (88) at (20, 2) {};
		\node [style=none] (89) at (20, -2.75) {};
		\node [style=state] (90) at (17.75, -1.25) {\small$y'_{N\!-\!1}$};
		\node [style=state] (91) at (17.5, 1.5) {\small$y'_1$};
		\node [style=none] (92) at (20, 1) {$\vdots$};
		\node [style=none] (93) at (5.5, 1) {$\vdots$};
	\end{pgfonlayer}
	\begin{pgfonlayer}{edgelayer}
		\node [style=boxLargerThin, fill=gray!25] (76) at (-13.5, 0) {};
		\node [style=boxLargerThin, fill=gray!25] (77) at (-8, 0) {};
		\draw [style=dashed] (5.center) to (15.center);
		\draw [style=dashed] (4.center) to (11.center);
		\draw [style=dashed, in=180, out=0] (0.center) to (7.center);
		\draw [style=-, line width=4 pt, draw=white, in=180, out=0] (9.center) to (6.center);
		\draw [style=-] (9.center) to (6.center);
		\draw [style=-, line width=4 pt, draw=white, in=180, out=0] (8.center) to (1.center);
		\draw [style=-] (8.center) to (1.center);
		\draw [style=-, line width=4 pt, draw=white, in=180, out=0] (14.center) to (3.center);
		\draw [style=-] (14.center) to (3.center);
		\draw [style=dashed] (17.center) to (16.center);
		\draw [style=dashed] (19.center) to (18.center);
		\draw [style=dashed] (21.center) to (20.center);
		\draw [style=-, line width=4 pt, draw=white, in=180, out=0] (27.center) to (26.center);
		\draw [style=-] (27.center) to (26.center);
		\draw [style=-, line width=4 pt, draw=white, in=180, out=0] (23.center) to (22.center);
		\draw [style=-] (23.center) to (22.center);
		\draw [style=-, line width=4 pt, draw=white, in=180, out=0] (25.center) to (24.center);
		\draw [style=-] (25.center) to (24.center);
		\draw [style=dashed] (29.center) to (28.center);
		\draw [style=dashed] (31.center) to (30.center);
		\draw [style=dashed] (33.center) to (32.center);
		\draw [style=->, draw=gray, in=-90, out=90] (34) to (35.center);
		\draw [style=dashed, in=180, out=0] (36.center) to (37.center);
		\draw [style=dashed] (39.center) to (38.center);
		\draw [style=dashed] (41.center) to (40.center);
		\draw [style=->, draw=gray, in=90, out=-90, looseness=0.75] (43) to (42.center);
		\draw [style=->, draw=gray, in=90, out=-90, looseness=0.75] (43) to (44.center);
		\draw [style=dashed, in=180, out=0] (48) to (62.center);
		\draw [style=dashed, in=180, out=0] (55) to (60.center);
		\draw [style=dashed, in=180, out=0] (70) to (59.center);
		\draw [style=-, line width=4 pt, draw=white, in=180, out=0] (64.center) to (50.center);
		\draw [style=-] (64.center) to (50.center);
		\draw [style=-, line width=4 pt, draw=white, in=180, out=0] (63.center) to (51.center);
		\draw [style=-] (63.center) to (51.center);
		\draw [style=-, line width=4 pt, draw=white, in=180, out=0] (49.center) to (67.center);
		\draw [style=-] (49.center) to (67.center);
		\draw [style=-, line width=4 pt, draw=white, in=180, out=0] (75.center) to (54.center);
		\draw [style=-] (75.center) to (54.center);
		\draw [style=-, line width=4 pt, draw=white, in=180, out=0] (57.center) to (53.center);
		\draw [style=-] (57.center) to (53.center);
		\draw [style=-, line width=4 pt, draw=white, in=180, out=0] (68.center) to (69.center);
		\draw [style=-] (68.center) to (69.center);
		\draw [style=dashed, in=180, out=0] (72.center) to (74);
		\draw [style=dashed, in=180, out=0] (73.center) to (52);
		\draw [style=dashed, in=180, out=0] (71.center) to (61);
		\draw [style=dashed, in=180, out=0] (83.center) to (80);
		\draw [style=dashed, in=180, out=0] (84.center) to (81);
		\draw [style=dashed, in=180, out=0] (85.center) to (82);
		\draw [style=dashed, in=180, out=0] (91) to (88.center);
		\draw [style=dashed, in=180, out=0] (90) to (86.center);
		\draw [style=dashed, in=180, out=0] (87) to (89.center);
	\end{pgfonlayer}
\end{tikzpicture}}
\end{equation}
A product family $M_1 \otimes ... \otimes M_N$ of processes---the local instruments---is sandwiched between two global classical processes $\sum_{\underline{x}',\underline{x}} p_{\underline{x}',\underline{x}} \ket{x_1...x_N} \bra{x'_1...x'_N}$ and $\sum_{\underline{y},\underline{y}'} q_{\underline{y},\underline{y}'}  \ket{y'_1...y'_N} \bra{y_1...y_N}$---the global classical operations---which we don't label explicitly for reasons of notational convenience. The global classical processes are allowed to act on the classical inputs and classical outputs of the local instruments, but leave all other inputs and outputs invariant---a fact which we denote by drawing the non-classical wires ``passing overhead''. Note that using global classical operations is equivalent to allowing classical communication between the parties: we can always implement one such global classical operation by allowing all parties to communicate their respective classical inputs to a distinguished party, who then performs the operation and sends the classical outputs back to the respective parties.

\paragraph{Purity.} Purity is a feature arising at the interface between quantum theory and thermodynamics~\cite{Chiribella2010,Chiribella2015}: pure processes can broadly be interpreted as not involving any probabilistic mixing due to non-trivial interactions with a discarded environment. To be precise, a process is pure if it cannot be decomposed in any way as a non-trivial $\reals^+$-linear combination of other processes:
\begin{equation}\label{purity}
	\resizebox{0.9\textwidth}{!}{\begin{tikzpicture}
	\begin{pgfonlayer}{nodelayer}
		\node [style=box] (0) at (-6, 0) {$M'$};
		\node [style=trace] (1) at (-3.75, -1) {};
		\node [style=none] (2) at (-5.5, -0.25) {};
		\node [style=tightlabelnode] (3) at (-3, 0.25) {};
		\node [style=tightlabelnode] (4) at (-8, 0) {};
		\node [style=none] (5) at (-5.5, 0.25) {};
		\node [style=none] (6) at (-6.5, 0) {};
		\node [style=tightlabelnode] (7) at (-9.5, 0) {$=$};
		\node [style=none] (8) at (-12.5, 0) {};
		\node [style=none] (9) at (-13.5, 0) {};
		\node [style=tightlabelnode] (10) at (-11, 0) {};
		\node [style=box] (11) at (-13, 0) {$M$};
		\node [style=tightlabelnode] (12) at (-15, 0) {};
		\node [style=tightlabelnode] (13) at (-0.5, 0) {$\Rightarrow$};
		\node [style=box] (14) at (11.5, 0.25) {$M$};
		\node [style=tightlabelnode] (15) at (9.5, 0.25) {};
		\node [style=tightlabelnode] (17) at (6.5, -1) {};
		\node [style=none] (18) at (4.5, 0.25) {};
		\node [style=tightlabelnode] (19) at (8, 0) {$=$};
		\node [style=none] (20) at (12, 0.25) {};
		\node [style=box] (21) at (4, 0) {$M'$};
		\node [style=tightlabelnode] (22) at (13.5, 0.25) {};
		\node [style=none] (23) at (3.5, 0) {};
		\node [style=tightlabelnode] (24) at (6.5, 0.25) {};
		\node [style=tightlabelnode] (25) at (2, 0) {};
		\node [style=none] (26) at (4.5, -0.25) {};
		\node [style=none] (27) at (11, 0.25) {};
		\node [style=state] (28) at (11.5, -1) {$p$};
		\node [style=tightlabelnode] (29) at (13.5, -1) {};
	\end{pgfonlayer}
	\begin{pgfonlayer}{edgelayer}
		\draw [style=-] (3) to (5.center);
		\draw [style=-] (4) to (6.center);
		\draw [style=dashed, in=180, out=0] (2.center) to (1);
		\draw [style=-] (10) to (8.center);
		\draw [style=-] (12) to (9.center);
		\draw [style=-] (24) to (18.center);
		\draw [style=-] (25) to (23.center);
		\draw [style=dashed, in=180, out=0] (26.center) to (17);
		\draw [style=-] (22) to (20.center);
		\draw [style=-] (15) to (27.center);
		\draw [style=dashed, in=180, out=0] (28) to (29);
	\end{pgfonlayer}
\end{tikzpicture}}
\end{equation}
It should be noted that purity is a notion that applies to both normalised and unnormalised processes: it is simply a statement of extremality in the convex cone of processes from a fixed input system to a fixed output system. When the process is normalised, the $\reals^+$-valued coefficients will sum to 1, and $\reals^+$-linear combinations reduce to the usual notion of probabilistic mixtures.

\paragraph{Time-reversal.} From a compositional perspective, the action of a generic notion of time-reversal can be described as some way of sending processes to other processes which have inputs and outputs swapped, while at the same time respecting their sequential/parallel compositional structure and their probabilistic structure. Categorically, this is captured by asking that time-reversal is a contravariant $\reals^+$-linear monoidal functor on the CPT, i.e. a $\reals^+$-linear dagger functor (in string diagram language) which coincides with the transpose on classical processes. Extremely relevant to this work is the duality established by any such notion of time-reversal between normalised processes and unital processes. In a CPT with time-reversal (as described above), a \emph{unital} process is one satisfying the following equation (where $M^\dagger$ is the time-reverse of $M$):
\begin{equation}\label{processUnital}
	\begin{tikzpicture}
	\begin{pgfonlayer}{nodelayer}
		\node [style=none] (1) at (-5, 0.25) {};
		\node [style=box] (3) at (-4.5, 0) {$M$};
		\node [style=none] (4) at (-5, -0.25) {};
		\node [style=none] (5) at (-4, 0.25) {};
		\node [style=none] (6) at (-4, -0.25) {};
		\node [style=none] (7) at (0, 0) {$=$};
		\node [style=none] (10) at (3.5, -0.25) {};
		\node [style=none] (11) at (-2.5, 0.25) {};
		\node [style=none] (12) at (-2.5, -0.25) {};
		\node [style=none] (13) at (3.5, 0.25) {};
	\end{pgfonlayer}
	\begin{pgfonlayer}{edgelayer}
		\node [style=cotrace] (0) at (-7.5, -0.25) {};
		\node [style=cotrace] (2) at (-6.5, 0.25) {};
		\node [style=cotrace] (8) at (1, -0.25) {};
		\node [style=cotrace] (9) at (2, 0.25) {};
		\draw [style=-] (2) to (1.center);
		\draw [style=-, in=180, out=0] (5.center) to (11.center);
		\draw [style=-, in=180, out=0] (9) to (13.center);
		\draw [style=-, line width=4 pt, draw=white] (6.center) to (12.center);
		\draw [style=dashed, in=180, out=0] (6.center) to (12.center);
		\draw [style=-, line width=4 pt, draw=white] (8) to (10.center);
		\draw [style=dashed, in=180, out=0] (8) to (10.center);
		\draw [style=-, line width=4 pt, draw=white] (0) to (4.center);
		\draw [style=dashed] (0) to (4.center);
	\end{pgfonlayer}
\end{tikzpicture}
\end{equation}
The reversed discarding maps are (up to proportionality factor) the uniform probability distribution (on classical systems) and the maximally mixed state (on generic systems). The relationship between normalised and unital processes established by time-reversal can then be summarised as follows:
\begin{equation}\label{processNormalisedAndDagger}
	\resizebox{0.9\textwidth}{!}{$
		\begin{tikzpicture}
	\begin{pgfonlayer}{nodelayer}
		\node [style=none] (0) at (-7.5, -0.25) {};
		\node [style=none] (1) at (-6, 0.25) {};
		\node [style=none] (2) at (-7.5, 0.25) {};
		\node [style=box] (3) at (-5.5, 0) {$M$};
		\node [style=none] (4) at (-6, -0.25) {};
		\node [style=none] (6) at (-5, 0.25) {};
		\node [style=none] (8) at (-5, -0.25) {};
		\node [style=none] (9) at (0, 0) {$=$};
		\node [style=none] (11) at (1, -0.25) {};
		\node [style=none] (12) at (1, 0.25) {};
	\end{pgfonlayer}
	\begin{pgfonlayer}{edgelayer}
		\node [style=trace] (13) at (3.5, -0.25) {};
		\node [style=trace] (5) at (-3.5, 0.25) {};
		\node [style=trace] (7) at (-2.5, -0.25) {};
		\node [style=trace] (10) at (2.5, 0.25) {};
		\draw [style=dashed] (0.center) to (4.center);
		\draw [style=-] (2.center) to (1.center);
		\draw [style=-, in=180, out=0] (6.center) to (5);
		\draw [style=-, in=180, out=0] (12.center) to (10);
		\draw [style=-, line width=4pt, draw=white] (8.center) to (7);
		\draw [style=-, line width=4pt, draw=white] (11.center) to (13);
		\draw [style=dashed, in=180, out=0] (8.center) to (7);
		\draw [style=dashed, in=180, out=0] (11.center) to (13);
	\end{pgfonlayer}
\end{tikzpicture}
		\hspace{10mm}\Leftrightarrow\hspace{10mm}
		\begin{tikzpicture}
	\begin{pgfonlayer}{nodelayer}
		\node [style=none] (1) at (-5, 0.25) {};
		\node [style=box] (3) at (-4.5, 0) {$M^\dagger$};
		\node [style=none] (4) at (-5, -0.25) {};
		\node [style=none] (5) at (-4, 0.25) {};
		\node [style=none] (6) at (-4, -0.25) {};
		\node [style=none] (7) at (0, 0) {$=$};
		\node [style=none] (10) at (3.5, -0.25) {};
		\node [style=none] (11) at (-2.5, 0.25) {};
		\node [style=none] (12) at (-2.5, -0.25) {};
		\node [style=none] (13) at (3.5, 0.25) {};
	\end{pgfonlayer}
	\begin{pgfonlayer}{edgelayer}
		\node [style=cotrace] (0) at (-7.5, -0.25) {};
		\node [style=cotrace] (2) at (-6.5, 0.25) {};
		\node [style=cotrace] (8) at (1, -0.25) {};
		\node [style=cotrace] (9) at (2, 0.25) {};
		\draw [style=-] (2) to (1.center);
		\draw [style=-, in=180, out=0] (5.center) to (11.center);
		\draw [style=-, in=180, out=0] (9) to (13.center);
		\draw [style=-, line width=4 pt, draw=white] (6.center) to (12.center);
		\draw [style=dashed, in=180, out=0] (6.center) to (12.center);
		\draw [style=-, line width=4 pt, draw=white] (8) to (10.center);
		\draw [style=dashed, in=180, out=0] (8) to (10.center);
		\draw [style=-, line width=4 pt, draw=white] (0) to (4.center);
		\draw [style=dashed] (0) to (4.center);
	\end{pgfonlayer}
\end{tikzpicture}
	$}
\end{equation}
In particular, the normalised processes which are invariant under time-reversal are exactly those which are both normalised and unital, such as the RaRe introduced by \cite{Chiribella2015}.

\section{Main result}

\paragraph{The distinguishing task.} Consider the following game between $N$ players, each player $j$ being in possession of a finite-dimensional quantum system $\SpaceH_j$. The players share a pure quantum state $\psi_b$, which they are guaranteed to be drawn at random from an orthonormal basis $(\psi_b)_{b \in B}$ of the joint quantum system $\bigotimes_{j=1}^N \SpaceH_j$. Here $B$ is a finite set of labels, with cardinality $|B| = \prod_{j=1}^N \dim{\SpaceH_j}$). They are tasked to identify the state with certainty using only local operations and classical communication. Our main result will state that the presence of even a single entangled state in the basis makes the task certainly impossible. Note that the task might be impossible even when the basis states are product states, cf. \cite{Betal99}.

\begin{theorem}[\textbf{Time-reversal and purity deny entanglement}]\label{thm_TRandPdenyE}
	If the players can succeed in the state distinguishing task by using any LOCC protocol (possibly involving unnormalised instruments/operations), then the complete orthonormal family $(\psi_b)_{b \in B}$ cannot contain any entangled states. Conversely, if the complete orthonormal family $(\psi_b)_{b \in B}$ contains only product states, then the players can succeed in the state distinguishing task using a unital LOCC protocol, although the latter might involve unnormalised instruments/operations.
\end{theorem}
\begin{proof}
We begin by considering a generic LOCC protocol (possibly involving unnormalised instruments/operations) that the player might be using to accomplish their distinguishing task. One such protocol would include some number $R>0$ of rounds, each round $r$ consisting of a global classical operation shared amongst the parties, followed by local instruments $M_i^{(r)}$ performed by the individual parties $i=1,...,N$. For the first round, we just assume this to a be a shared global classical state, with no inputs. After the last round, a global post-processing operation is applied to the classical outputs of the local instruments, producing a classical output in the set $B$ which identifies the state that the players believe they were given. Overall, the protocol takes the following shape:
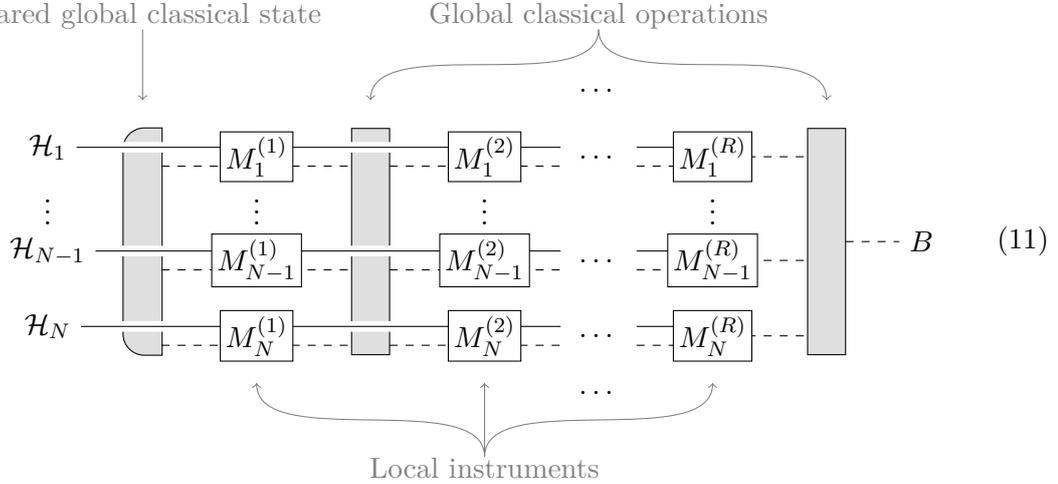
\begin{equation}\label{LOCCprotocol}
	\begin{tikzpicture}
	\begin{pgfonlayer}{nodelayer}
		\node [style=box] (0) at (-5.5, 2.25) {$M_{1}^{(1)}$};
		\node [style=box] (1) at (-5.5, -0.5) {$M_{N-1}^{(1)}$};
		\node [style=box] (2) at (-5.5, -2.5) {$M_{N}^{(1)}$};
		\node [style=none] (3) at (-8, 2) {};
		\node [style=none] (4) at (-6, 2) {};
		\node [style=none] (5) at (-6, -0.75) {};
		\node [style=none] (6) at (-8, -0.75) {};
		\node [style=none] (7) at (-6, -2.75) {};
		\node [style=none] (8) at (-8, -2.75) {};
		\node [style=labelnode] (9) at (-11, 2.5) {$\SpaceH_1$};
		\node [style=none] (10) at (-6, 2.5) {};
		\node [style=none] (11) at (-6, -0.25) {};
		\node [style=labelnode] (12) at (-11, -0.25) {$\SpaceH_{N-1}$};
		\node [style=none] (13) at (-6, -2.25) {};
		\node [style=labelnode] (14) at (-11, -2.25) {$\SpaceH_{N}$};
		\node [style=none] (15) at (-5.5, 1) {$\vdots$};
		\node [style=none] (16) at (-2, -2.75) {};
		\node [style=none] (17) at (0, -0.25) {};
		\node [style=box] (18) at (0.5, -2.5) {$M_{N}^{(2)}$};
		\node [style=none] (19) at (0, -2.25) {};
		\node [style=none] (20) at (-2, -0.75) {};
		\node [style=none] (21) at (-2, 2) {};
		\node [style=none] (22) at (0, 2.5) {};
		\node [style=none] (23) at (0, -2.75) {};
		\node [style=none] (24) at (-5, -0.25) {};
		\node [style=none] (25) at (-5, 2.5) {};
		\node [style=box] (26) at (0.5, -0.5) {$M_{N-1}^{(2)}$};
		\node [style=none] (27) at (0, -0.75) {};
		\node [style=none] (28) at (0.5, 1) {$\vdots$};
		\node [style=box] (29) at (0.5, 2.25) {$M_{1}^{(2)}$};
		\node [style=none] (30) at (-5, -2.25) {};
		\node [style=none] (31) at (0, 2) {};
		\node [style=none] (32) at (-3, 2) {};
		\node [style=none] (33) at (-5, 2) {};
		\node [style=none] (34) at (-3, -0.75) {};
		\node [style=none] (35) at (-5, -0.75) {};
		\node [style=none] (36) at (-3, -2.75) {};
		\node [style=none] (37) at (-5, -2.75) {};
		\node [style=none] (38) at (2.5, -0.25) {};
		\node [style=none] (39) at (1, -0.25) {};
		\node [style=none] (40) at (2.5, -2.25) {};
		\node [style=none] (41) at (1, -2.25) {};
		\node [style=none] (42) at (2.5, 2.5) {};
		\node [style=none] (43) at (1, 2.5) {};
		\node [style=none] (44) at (2.5, 2) {};
		\node [style=none] (45) at (1, 2) {};
		\node [style=none] (46) at (2.5, -0.75) {};
		\node [style=none] (47) at (1, -0.75) {};
		\node [style=none] (48) at (2.5, -2.75) {};
		\node [style=none] (49) at (1, -2.75) {};
		\node [style=none] (50) at (3.5, 2.25) {$\dots$};
		\node [style=none] (51) at (3.5, -0.5) {$\dots$};
		\node [style=none] (52) at (3.5, -2.5) {$\dots$};
		\node [style=none] (53) at (4.5, 2.5) {};
		\node [style=none] (54) at (6, -2.25) {};
		\node [style=none] (55) at (4.5, -2.75) {};
		\node [style=none] (56) at (6, -2.75) {};
		\node [style=none] (57) at (6, 2) {};
		\node [style=none] (58) at (6, -0.75) {};
		\node [style=none] (59) at (4.5, -0.25) {};
		\node [style=none] (60) at (4.5, -0.75) {};
		\node [style=none] (61) at (6, -0.25) {};
		\node [style=none] (62) at (6, 2.5) {};
		\node [style=none] (63) at (4.5, -2.25) {};
		\node [style=none] (64) at (4.5, 2) {};
		\node [style=none] (65) at (9, 2.25) {};
		\node [style=none] (66) at (9, -2.5) {};
		\node [style=none] (67) at (9, -0.5) {};
		\node [style=box] (68) at (6.5, -0.5) {$M_{N-1}^{(R)}$};
		\node [style=none] (69) at (7, -0.5) {};
		\node [style=none] (70) at (7, -2.5) {};
		\node [style=box] (71) at (6.5, -2.5) {$M_{N}^{(R)}$};
		\node [style=box] (72) at (6.5, 2.25) {$M_{1}^{(R)}$};
		\node [style=none] (73) at (6.5, 1) {$\vdots$};
		\node [style=none] (74) at (7, 2.25) {};
		\node [style=labelnode] (75) at (12, 0) {$B$};
		\node [style=none] (76) at (10, 0) {};
		\node [style=none] (77) at (-11, 1) {$\vdots$};
		\node [style=none] (78) at (9.5, 3.75) {};
		\node [style=tightlabelnode, text=gray] (79) at (3.5, 6) {Global classical operations};
		\node [style=none] (80) at (-2.5, 3.75) {};
		\node [style=none] (81) at (3.5, 4) {$\dots$};
		\node [style=tightlabelnode, text=gray] (82) at (0.5, -6) {Local instruments};
		\node [style=none] (83) at (0.5, -3.75) {};
		\node [style=none] (84) at (6.5, -3.75) {};
		\node [style=none] (85) at (-5.5, -3.75) {};
		\node [style=none] (86) at (3.5, -4) {$\dots$};
		\node [style=tightlabelnode, text=gray] (87) at (-8.5, 6) {Shared global classical state};
		\node [style=none] (88) at (-8.5, 3.75) {};
	\end{pgfonlayer}
	\begin{pgfonlayer}{edgelayer}
		\node [style=boxLargerThin, fill=gray!25] (89) at (9.5, 0) {};
		\node [style=stateLargerThin, fill=gray!25] (90) at (-8.5, 0) {};
		\node [style=boxLargerThin, fill=gray!25] (91) at (-2.5, 0) {};
		\draw [style=dashed] (3.center) to (4.center);
		\draw [style=dashed] (6.center) to (5.center);
		\draw [style=dashed, in=180, out=0] (8.center) to (7.center);
		\draw [style=-, line width=4 pt, draw=white, in=180, out=0] (9) to (10.center);
		\draw [style=-] (9) to (10.center);
		\draw [style=-, line width=4 pt, draw=white, in=180, out=0] (12) to (11.center);
		\draw [style=-] (12) to (11.center);
		\draw [style=-, line width=4 pt, draw=white, in=180, out=0] (14) to (13.center);
		\draw [style=-] (14) to (13.center);
		\draw [style=dashed] (21.center) to (31.center);
		\draw [style=dashed] (20.center) to (27.center);
		\draw [style=dashed, in=180, out=0] (16.center) to (23.center);
		\draw [style=-, line width=4 pt, draw=white, in=180, out=0] (25.center) to (22.center);
		\draw [style=-] (25.center) to (22.center);
		\draw [style=-, line width=4 pt, draw=white, in=180, out=0] (24.center) to (17.center);
		\draw [style=-] (24.center) to (17.center);
		\draw [style=-, line width=4 pt, draw=white, in=180, out=0] (30.center) to (19.center);
		\draw [style=-] (30.center) to (19.center);
		\draw [style=dashed] (33.center) to (32.center);
		\draw [style=dashed] (35.center) to (34.center);
		\draw [style=dashed] (37.center) to (36.center);
		\draw [style=-, line width=4 pt, draw=white, in=180, out=0] (43.center) to (42.center);
		\draw [style=-] (43.center) to (42.center);
		\draw [style=-, line width=4 pt, draw=white, in=180, out=0] (39.center) to (38.center);
		\draw [style=-] (39.center) to (38.center);
		\draw [style=-, line width=4 pt, draw=white, in=180, out=0] (41.center) to (40.center);
		\draw [style=-] (41.center) to (40.center);
		\draw [style=dashed] (45.center) to (44.center);
		\draw [style=dashed] (47.center) to (46.center);
		\draw [style=dashed] (49.center) to (48.center);
		\draw [style=-, line width=4 pt, draw=white, in=180, out=0] (53.center) to (62.center);
		\draw [style=-] (53.center) to (62.center);
		\draw [style=-, line width=4 pt, draw=white, in=180, out=0] (59.center) to (61.center);
		\draw [style=-] (59.center) to (61.center);
		\draw [style=-, line width=4 pt, draw=white, in=180, out=0] (63.center) to (54.center);
		\draw [style=-] (63.center) to (54.center);
		\draw [style=dashed] (64.center) to (57.center);
		\draw [style=dashed] (60.center) to (58.center);
		\draw [style=dashed] (55.center) to (56.center);
		\draw [style=dashed] (74.center) to (65.center);
		\draw [style=dashed] (69.center) to (67.center);
		\draw [style=dashed] (70.center) to (66.center);
		\draw [style=dashed] (76.center) to (75);
		\draw [style=->, draw=gray, in=90, out=-90, looseness=0.75] (79) to (78.center);
		\draw [style=->, draw=gray, in=90, out=-90, looseness=0.75] (79) to (80.center);
		\draw [style=->, draw=gray, in=-90, out=90, looseness=0.75] (82) to (85.center);
		\draw [style=->, draw=gray, in=-90, out=90] (82) to (83.center);
		\draw [style=->, draw=gray, in=-90, out=90, looseness=0.75] (82) to (84.center);
		\draw [style=->, draw=gray, in=90, out=-90, looseness=0.50] (87) to (88.center);
	\end{pgfonlayer}
\end{tikzpicture}
\end{equation}
The protocol implements the desired state distinguishing task if and only if the following is true for every choice $\psi_b$ of shared initial state in the complete orthogonal family $(\psi_b)_{b \in B}$:
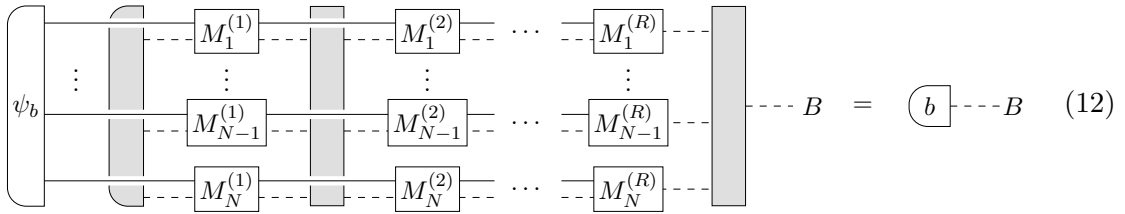
\begin{equation}\label{LOCCprotocolCondition}
	\resizebox{0.9\textwidth}{!}{\begin{tikzpicture}
	\begin{pgfonlayer}{nodelayer}
		\node [style=box] (0) at (-8, 2.25) {$M_{1}^{(1)}$};
		\node [style=box] (1) at (-8, -0.5) {$M_{N-1}^{(1)}$};
		\node [style=box] (2) at (-8, -2.5) {$M_{N}^{(1)}$};
		\node [style=none] (3) at (-10.5, 2) {};
		\node [style=none] (4) at (-8.5, 2) {};
		\node [style=none] (5) at (-8.5, -0.75) {};
		\node [style=none] (6) at (-10.5, -0.75) {};
		\node [style=none] (7) at (-8.5, -2.75) {};
		\node [style=none] (8) at (-10.5, -2.75) {};
		\node [style=none] (9) at (-13.5, 2.5) {};
		\node [style=none] (10) at (-8.5, 2.5) {};
		\node [style=none] (11) at (-8.5, -0.25) {};
		\node [style=none] (12) at (-13.5, -0.25) {};
		\node [style=none] (13) at (-8.5, -2.25) {};
		\node [style=none] (14) at (-13.5, -2.25) {};
		\node [style=none] (15) at (-8, 1) {$\vdots$};
		\node [style=none] (16) at (-4.5, -2.75) {};
		\node [style=none] (17) at (-2.5, -0.25) {};
		\node [style=box] (18) at (-2, -2.5) {$M_{N}^{(2)}$};
		\node [style=none] (19) at (-2.5, -2.25) {};
		\node [style=none] (20) at (-4.5, -0.75) {};
		\node [style=none] (21) at (-4.5, 2) {};
		\node [style=none] (22) at (-2.5, 2.5) {};
		\node [style=none] (23) at (-2.5, -2.75) {};
		\node [style=none] (24) at (-7.5, -0.25) {};
		\node [style=none] (25) at (-7.5, 2.5) {};
		\node [style=box] (26) at (-2, -0.5) {$M_{N-1}^{(2)}$};
		\node [style=none] (27) at (-2.5, -0.75) {};
		\node [style=none] (28) at (-2, 1) {$\vdots$};
		\node [style=box] (29) at (-2, 2.25) {$M_{1}^{(2)}$};
		\node [style=none] (30) at (-7.5, -2.25) {};
		\node [style=none] (31) at (-2.5, 2) {};
		\node [style=none] (32) at (-5.5, 2) {};
		\node [style=none] (33) at (-7.5, 2) {};
		\node [style=none] (34) at (-5.5, -0.75) {};
		\node [style=none] (35) at (-7.5, -0.75) {};
		\node [style=none] (36) at (-5.5, -2.75) {};
		\node [style=none] (37) at (-7.5, -2.75) {};
		\node [style=none] (38) at (0, -0.25) {};
		\node [style=none] (39) at (-1.5, -0.25) {};
		\node [style=none] (40) at (0, -2.25) {};
		\node [style=none] (41) at (-1.5, -2.25) {};
		\node [style=none] (42) at (0, 2.5) {};
		\node [style=none] (43) at (-1.5, 2.5) {};
		\node [style=none] (44) at (0, 2) {};
		\node [style=none] (45) at (-1.5, 2) {};
		\node [style=none] (46) at (0, -0.75) {};
		\node [style=none] (47) at (-1.5, -0.75) {};
		\node [style=none] (48) at (0, -2.75) {};
		\node [style=none] (49) at (-1.5, -2.75) {};
		\node [style=none] (50) at (1, 2.25) {$\dots$};
		\node [style=none] (51) at (1, -0.5) {$\dots$};
		\node [style=none] (52) at (1, -2.5) {$\dots$};
		\node [style=none] (53) at (2, 2.5) {};
		\node [style=none] (54) at (3.5, -2.25) {};
		\node [style=none] (55) at (2, -2.75) {};
		\node [style=none] (56) at (3.5, -2.75) {};
		\node [style=none] (57) at (3.5, 2) {};
		\node [style=none] (58) at (3.5, -0.75) {};
		\node [style=none] (59) at (2, -0.25) {};
		\node [style=none] (60) at (2, -0.75) {};
		\node [style=none] (61) at (3.5, -0.25) {};
		\node [style=none] (62) at (3.5, 2.5) {};
		\node [style=none] (63) at (2, -2.25) {};
		\node [style=none] (64) at (2, 2) {};
		\node [style=none] (65) at (6.5, 2.25) {};
		\node [style=none] (66) at (6.5, -2.5) {};
		\node [style=none] (67) at (6.5, -0.5) {};
		\node [style=box] (68) at (4, -0.5) {$M_{N-1}^{(R)}$};
		\node [style=none] (69) at (4.5, -0.5) {};
		\node [style=none] (70) at (4.5, -2.5) {};
		\node [style=box] (71) at (4, -2.5) {$M_{N}^{(R)}$};
		\node [style=box] (72) at (4, 2.25) {$M_{1}^{(R)}$};
		\node [style=none] (73) at (4, 1) {$\vdots$};
		\node [style=none] (74) at (4.5, 2.25) {};
		\node [style=labelnode] (75) at (9.5, 0) {$B$};
		\node [style=none] (76) at (7.5, 0) {};
		\node [style=none] (77) at (-12.5, 1) {$\vdots$};
		\node [style=stateLargerThin] (78) at (-14, 0) {$\psi_b$};
		\node [style=labelnode] (79) at (11, 0) {$=$};
		\node [style=state] (80) at (13, 0) {$b$};
		\node [style=tightlabelnode] (81) at (15.5, 0) {$B$};
	\end{pgfonlayer}
	\begin{pgfonlayer}{edgelayer}
		\node [style=boxLargerThin, fill=gray!25] (82) at (7, 0) {};
		\node [style=stateLargerThin, fill=gray!25] (83) at (-11, 0) {};
		\node [style=boxLargerThin, fill=gray!25] (84) at (-5, 0) {};
		\draw [style=dashed] (3.center) to (4.center);
		\draw [style=dashed] (6.center) to (5.center);
		\draw [style=dashed, in=180, out=0] (8.center) to (7.center);
		\draw [style=-, line width=4 pt, draw=white, in=180, out=0] (9.center) to (10.center);
		\draw [style=-] (9.center) to (10.center);
		\draw [style=-, line width=4 pt, draw=white, in=180, out=0] (12.center) to (11.center);
		\draw [style=-] (12.center) to (11.center);
		\draw [style=-, line width=4 pt, draw=white, in=180, out=0] (14.center) to (13.center);
		\draw [style=-] (14.center) to (13.center);
		\draw [style=dashed] (21.center) to (31.center);
		\draw [style=dashed] (20.center) to (27.center);
		\draw [style=dashed, in=180, out=0] (16.center) to (23.center);
		\draw [style=-, line width=4 pt, draw=white, in=180, out=0] (25.center) to (22.center);
		\draw [style=-] (25.center) to (22.center);
		\draw [style=-, line width=4 pt, draw=white, in=180, out=0] (24.center) to (17.center);
		\draw [style=-] (24.center) to (17.center);
		\draw [style=-, line width=4 pt, draw=white, in=180, out=0] (30.center) to (19.center);
		\draw [style=-] (30.center) to (19.center);
		\draw [style=dashed] (33.center) to (32.center);
		\draw [style=dashed] (35.center) to (34.center);
		\draw [style=dashed] (37.center) to (36.center);
		\draw [style=-, line width=4 pt, draw=white, in=180, out=0] (43.center) to (42.center);
		\draw [style=-] (43.center) to (42.center);
		\draw [style=-, line width=4 pt, draw=white, in=180, out=0] (39.center) to (38.center);
		\draw [style=-] (39.center) to (38.center);
		\draw [style=-, line width=4 pt, draw=white, in=180, out=0] (41.center) to (40.center);
		\draw [style=-] (41.center) to (40.center);
		\draw [style=dashed] (45.center) to (44.center);
		\draw [style=dashed] (47.center) to (46.center);
		\draw [style=dashed] (49.center) to (48.center);
		\draw [style=-, line width=4 pt, draw=white, in=180, out=0] (53.center) to (62.center);
		\draw [style=-] (53.center) to (62.center);
		\draw [style=-, line width=4 pt, draw=white, in=180, out=0] (59.center) to (61.center);
		\draw [style=-] (59.center) to (61.center);
		\draw [style=-, line width=4 pt, draw=white, in=180, out=0] (63.center) to (54.center);
		\draw [style=-] (63.center) to (54.center);
		\draw [style=dashed] (64.center) to (57.center);
		\draw [style=dashed] (60.center) to (58.center);
		\draw [style=dashed] (55.center) to (56.center);
		\draw [style=dashed] (74.center) to (65.center);
		\draw [style=dashed] (69.center) to (67.center);
		\draw [style=dashed] (70.center) to (66.center);
		\draw [style=dashed] (76.center) to (75);
		\draw [style=dashed] (80) to (81);
	\end{pgfonlayer}
\end{tikzpicture}}
\end{equation}
Diagram \ref{LOCCprotocolCondition} above shows the situation in which the players apply their $R$-round LOCC protocol to the initial state $\psi_b$: the correctness requirement for the protocol is captured by the fact that this process deterministically results in the correct label $b \in B$. Now we consider the time-reversal of the entire protocol used by the players, i.e. we take the dagger of all global classical operations (which are simply transposed) and local instruments (which are transformed in some other way, depending on the exact choice of dagger functor implementing the time-reversal):
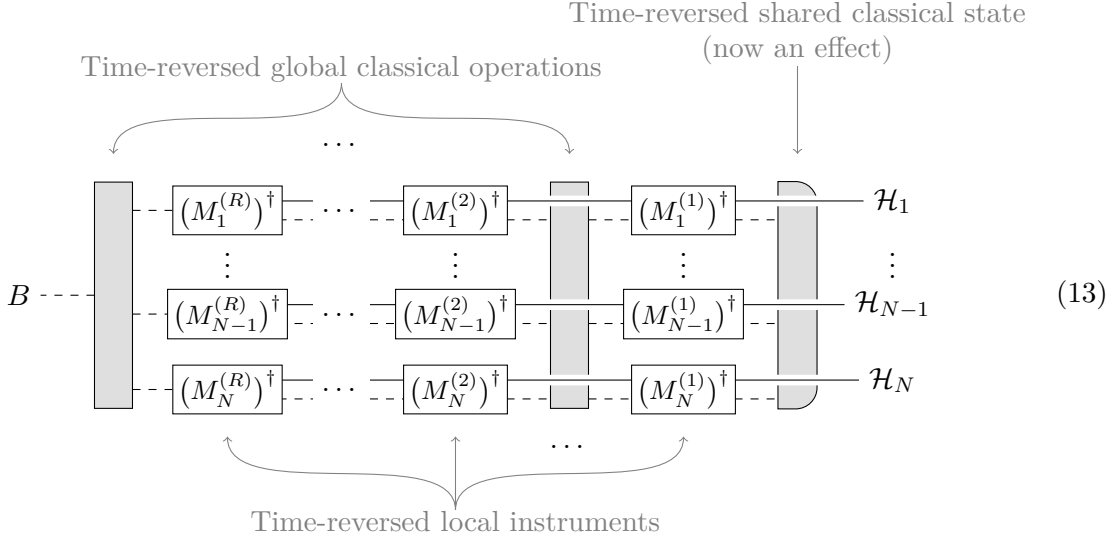
\begin{equation}\label{LOCCprotocolTimeReversed}
	\begin{tikzpicture}
	\begin{pgfonlayer}{nodelayer}
		\node [style=box] (0) at (4, 2.25) {\small$\big(M_{1}^{(1)}\big)^\dagger$};
		\node [style=box] (1) at (4, -0.5) {\small$\big(M_{N-1}^{(1)}\big)^\dagger$};
		\node [style=box] (2) at (4, -2.5) {\small$\big(M_{N}^{(1)}\big)^\dagger$};
		\node [style=none] (3) at (4.5, 2) {};
		\node [style=none] (4) at (6.5, 2) {};
		\node [style=none] (5) at (6.5, -0.75) {};
		\node [style=none] (6) at (4.5, -0.75) {};
		\node [style=none] (7) at (6.5, -2.75) {};
		\node [style=none] (8) at (4.5, -2.75) {};
		\node [style=labelnode] (9) at (9.5, 2.5) {$\SpaceH_1$};
		\node [style=none] (10) at (4.5, 2.5) {};
		\node [style=none] (11) at (4.5, -0.25) {};
		\node [style=labelnode] (12) at (9.5, -0.25) {$\SpaceH_{N-1}$};
		\node [style=none] (13) at (4.5, -2.25) {};
		\node [style=labelnode] (14) at (9.5, -2.25) {$\SpaceH_{N}$};
		\node [style=none] (15) at (4, 1) {$\vdots$};
		\node [style=none] (16) at (1.5, -2.75) {};
		\node [style=none] (17) at (3.5, -0.25) {};
		\node [style=box] (18) at (-2, -2.5) {\small$\big(M_{N}^{(2)}\big)^\dagger$};
		\node [style=none] (19) at (3.5, -2.25) {};
		\node [style=none] (20) at (1.5, -0.75) {};
		\node [style=none] (21) at (1.5, 2) {};
		\node [style=none] (22) at (3.5, 2.5) {};
		\node [style=none] (23) at (3.5, -2.75) {};
		\node [style=none] (24) at (-1.5, -0.25) {};
		\node [style=none] (25) at (-1.5, 2.5) {};
		\node [style=box] (26) at (-2, -0.5) {\small$\big(M_{N-1}^{(2)}\big)^\dagger$};
		\node [style=none] (27) at (3.5, -0.75) {};
		\node [style=none] (28) at (-2, 1) {$\vdots$};
		\node [style=box] (29) at (-2, 2.25) {\small$\big(M_{1}^{(2)}\big)^\dagger$};
		\node [style=none] (30) at (-1.5, -2.25) {};
		\node [style=none] (31) at (3.5, 2) {};
		\node [style=none] (32) at (0.5, 2) {};
		\node [style=none] (33) at (-1.5, 2) {};
		\node [style=none] (34) at (0.5, -0.75) {};
		\node [style=none] (35) at (-1.5, -0.75) {};
		\node [style=none] (36) at (0.5, -2.75) {};
		\node [style=none] (37) at (-1.5, -2.75) {};
		\node [style=none] (38) at (-5.75, -0.25) {};
		\node [style=none] (39) at (-7.5, -0.25) {};
		\node [style=none] (40) at (-5.75, -2.25) {};
		\node [style=none] (41) at (-7.5, -2.25) {};
		\node [style=none] (42) at (-5.75, 2.5) {};
		\node [style=none] (43) at (-7.5, 2.5) {};
		\node [style=none] (44) at (-5.75, 2) {};
		\node [style=none] (45) at (-7.5, 2) {};
		\node [style=none] (46) at (-5.75, -0.75) {};
		\node [style=none] (47) at (-7.5, -0.75) {};
		\node [style=none] (48) at (-5.75, -2.75) {};
		\node [style=none] (49) at (-7.5, -2.75) {};
		\node [style=none] (50) at (-5, 2.25) {$\dots$};
		\node [style=none] (51) at (-5, -0.5) {$\dots$};
		\node [style=none] (52) at (-5, -2.5) {$\dots$};
		\node [style=none] (53) at (-4.25, 2.5) {};
		\node [style=none] (54) at (-2.5, -2.25) {};
		\node [style=none] (55) at (-4.25, -2.75) {};
		\node [style=none] (56) at (-2.5, -2.75) {};
		\node [style=none] (57) at (-2.5, 2) {};
		\node [style=none] (58) at (-2.5, -0.75) {};
		\node [style=none] (59) at (-4.25, -0.25) {};
		\node [style=none] (60) at (-4.25, -0.75) {};
		\node [style=none] (61) at (-2.5, -0.25) {};
		\node [style=none] (62) at (-2.5, 2.5) {};
		\node [style=none] (63) at (-4.25, -2.25) {};
		\node [style=none] (64) at (-4.25, 2) {};
		\node [style=none] (65) at (9.5, 1) {$\vdots$};
		\node [style=none] (66) at (-11.5, 0) {};
		\node [style=labelnode] (67) at (-13.5, 0) {$B$};
		\node [style=box] (68) at (-8, -0.5) {\small$\big(M_{N-1}^{(R)}\big)^\dagger$};
		\node [style=none] (69) at (-8, 1) {$\vdots$};
		\node [style=none] (70) at (-10.5, 2.25) {};
		\node [style=none] (71) at (-8.5, -2.5) {};
		\node [style=none] (72) at (-10.5, -0.5) {};
		\node [style=box] (73) at (-8, -2.5) {\small$\big(M_{N}^{(R)}\big)^\dagger$};
		\node [style=none] (74) at (-10.5, -2.5) {};
		\node [style=none] (75) at (-8.5, -0.5) {};
		\node [style=box] (76) at (-8, 2.25) {\small$\big(M_{1}^{(R)}\big)^\dagger$};
		\node [style=none] (77) at (-8.5, 2.25) {};
		\node [style=none] (78) at (1, 3.75) {};
		\node [style=tightlabelnode, text=gray] (79) at (-5, 6) {Time-reversed global classical operations};
		\node [style=none] (80) at (-11, 3.75) {};
		\node [style=none] (81) at (-5, 4) {$\dots$};
		\node [style=tightlabelnode, text=gray] (82) at (-2, -6) {Time-reversed local instruments};
		\node [style=none] (83) at (-2, -3.75) {};
		\node [style=none] (84) at (4, -3.75) {};
		\node [style=none] (85) at (-8, -3.75) {};
		\node [style=none] (86) at (1, -4) {$\dots$};
		\node [style=tightlabelnode, text=gray] (87) at (7, 6.5) {(now an effect)};
		\node [style=none] (88) at (7, 3.75) {};
		\node [style=tightlabelnode, text=gray] (89) at (7, 7.5) {Time-reversed shared classical state };
	\end{pgfonlayer}
	\begin{pgfonlayer}{edgelayer}
		\node [style=boxLargerThin, fill=gray!25] (90) at (1, 0) {};
		\node [style=boxLargerThin, fill=gray!25] (91) at (-11, 0) {};
		\node [style=effectLargerThin, fill=gray!25] (92) at (7, 0) {};
		\draw [style=dashed] (3.center) to (4.center);
		\draw [style=dashed] (6.center) to (5.center);
		\draw [style=dashed, in=180, out=0] (8.center) to (7.center);
		\draw [style=-, line width=4 pt, draw=white, in=0, out=180] (9) to (10.center);
		\draw [style=-] (9) to (10.center);
		\draw [style=-, line width=4 pt, draw=white, in=0, out=180] (12) to (11.center);
		\draw [style=-] (12) to (11.center);
		\draw [style=-, line width=4 pt, draw=white, in=0, out=180] (14) to (13.center);
		\draw [style=-] (14) to (13.center);
		\draw [style=dashed] (21.center) to (31.center);
		\draw [style=dashed] (20.center) to (27.center);
		\draw [style=dashed, in=180, out=0] (16.center) to (23.center);
		\draw [style=-, line width=4 pt, draw=white, in=180, out=0] (25.center) to (22.center);
		\draw [style=-] (25.center) to (22.center);
		\draw [style=-, line width=4 pt, draw=white, in=180, out=0] (24.center) to (17.center);
		\draw [style=-] (24.center) to (17.center);
		\draw [style=-, line width=4 pt, draw=white, in=180, out=0] (30.center) to (19.center);
		\draw [style=-] (30.center) to (19.center);
		\draw [style=dashed] (33.center) to (32.center);
		\draw [style=dashed] (35.center) to (34.center);
		\draw [style=dashed] (37.center) to (36.center);
		\draw [style=-, line width=4 pt, draw=white, in=180, out=0] (43.center) to (42.center);
		\draw [style=-] (43.center) to (42.center);
		\draw [style=-, line width=4 pt, draw=white, in=180, out=0] (39.center) to (38.center);
		\draw [style=-] (39.center) to (38.center);
		\draw [style=-, line width=4 pt, draw=white, in=180, out=0] (41.center) to (40.center);
		\draw [style=-] (41.center) to (40.center);
		\draw [style=dashed] (45.center) to (44.center);
		\draw [style=dashed] (47.center) to (46.center);
		\draw [style=dashed] (49.center) to (48.center);
		\draw [style=-, line width=4 pt, draw=white, in=180, out=0] (53.center) to (62.center);
		\draw [style=-] (53.center) to (62.center);
		\draw [style=-, line width=4 pt, draw=white, in=180, out=0] (59.center) to (61.center);
		\draw [style=-] (59.center) to (61.center);
		\draw [style=-, line width=4 pt, draw=white, in=180, out=0] (63.center) to (54.center);
		\draw [style=-] (63.center) to (54.center);
		\draw [style=dashed] (64.center) to (57.center);
		\draw [style=dashed] (60.center) to (58.center);
		\draw [style=dashed] (55.center) to (56.center);
		\draw [style=dashed] (66.center) to (67);
		\draw [style=dashed] (70.center) to (77.center);
		\draw [style=dashed] (72.center) to (75.center);
		\draw [style=dashed] (74.center) to (71.center);
		\draw [style=->, draw=gray, in=90, out=-90, looseness=0.75] (79) to (78.center);
		\draw [style=->, draw=gray, in=90, out=-90, looseness=0.75] (79) to (80.center);
		\draw [style=->, draw=gray, in=-90, out=90, looseness=0.75] (82) to (85.center);
		\draw [style=->, draw=gray, in=-90, out=90] (82) to (83.center);
		\draw [style=->, draw=gray, in=-90, out=90, looseness=0.75] (82) to (84.center);
		\draw [style=->, draw=gray, in=90, out=-90, looseness=0.50] (87) to (88.center);
	\end{pgfonlayer}
\end{tikzpicture}
\end{equation}
Note that the time-reversed LOCC protocol of Diagram \ref{LOCCprotocolTimeReversed} need not be normalised, even when the original LOCC protocol of Diagram \ref{LOCCprotocol} is normalised (in the latter case, however, the time-reversed protocol is certainly unital).  Under time-reversal, the protocol success condition of Equation \ref{LOCCprotocolCondition} turns into a state preparation condition:
\begin{equation}\label{LOCCprotocolTimeReversedCondition}
	\resizebox{0.9\textwidth}{!}{\begin{tikzpicture}
	\begin{pgfonlayer}{nodelayer}
		\node [style=box] (0) at (0, 2.25) {\small$\big(M_{1}^{(1)}\big)^\dagger$};
		\node [style=box] (1) at (0, -0.5) {\small$\big(M_{N-1}^{(1)}\big)^\dagger$};
		\node [style=box] (2) at (0, -2.5) {\small$\big(M_{N}^{(1)}\big)^\dagger$};
		\node [style=none] (3) at (0.5, 2) {};
		\node [style=none] (4) at (2.5, 2) {};
		\node [style=none] (5) at (2.5, -0.75) {};
		\node [style=none] (6) at (0.5, -0.75) {};
		\node [style=none] (7) at (2.5, -2.75) {};
		\node [style=none] (8) at (0.5, -2.75) {};
		\node [style=labelnode] (9) at (5.5, 2.5) {$\SpaceH_1$};
		\node [style=none] (10) at (0.5, 2.5) {};
		\node [style=none] (11) at (0.5, -0.25) {};
		\node [style=labelnode] (12) at (5.5, -0.25) {$\SpaceH_{N-1}$};
		\node [style=none] (13) at (0.5, -2.25) {};
		\node [style=labelnode] (14) at (5.5, -2.25) {$\SpaceH_{N}$};
		\node [style=none] (15) at (0, 1) {$\vdots$};
		\node [style=none] (16) at (-2.5, -2.75) {};
		\node [style=none] (17) at (-0.5, -0.25) {};
		\node [style=box] (18) at (-6, -2.5) {\small$\big(M_{N}^{(2)}\big)^\dagger$};
		\node [style=none] (19) at (-0.5, -2.25) {};
		\node [style=none] (20) at (-2.5, -0.75) {};
		\node [style=none] (21) at (-2.5, 2) {};
		\node [style=none] (22) at (-0.5, 2.5) {};
		\node [style=none] (23) at (-0.5, -2.75) {};
		\node [style=none] (24) at (-5.5, -0.25) {};
		\node [style=none] (25) at (-5.5, 2.5) {};
		\node [style=box] (26) at (-6, -0.5) {\small$\big(M_{N-1}^{(2)}\big)^\dagger$};
		\node [style=none] (27) at (-0.5, -0.75) {};
		\node [style=none] (28) at (-6, 1) {$\vdots$};
		\node [style=box] (29) at (-6, 2.25) {\small$\big(M_{1}^{(2)}\big)^\dagger$};
		\node [style=none] (30) at (-5.5, -2.25) {};
		\node [style=none] (31) at (-0.5, 2) {};
		\node [style=none] (32) at (-3.5, 2) {};
		\node [style=none] (33) at (-5.5, 2) {};
		\node [style=none] (34) at (-3.5, -0.75) {};
		\node [style=none] (35) at (-5.5, -0.75) {};
		\node [style=none] (36) at (-3.5, -2.75) {};
		\node [style=none] (37) at (-5.5, -2.75) {};
		\node [style=none] (38) at (-9.75, -0.25) {};
		\node [style=none] (39) at (-11.5, -0.25) {};
		\node [style=none] (40) at (-9.75, -2.25) {};
		\node [style=none] (41) at (-11.5, -2.25) {};
		\node [style=none] (42) at (-9.75, 2.5) {};
		\node [style=none] (43) at (-11.5, 2.5) {};
		\node [style=none] (44) at (-9.75, 2) {};
		\node [style=none] (45) at (-11.5, 2) {};
		\node [style=none] (46) at (-9.75, -0.75) {};
		\node [style=none] (47) at (-11.5, -0.75) {};
		\node [style=none] (48) at (-9.75, -2.75) {};
		\node [style=none] (49) at (-11.5, -2.75) {};
		\node [style=none] (50) at (-9, 2.25) {$\dots$};
		\node [style=none] (51) at (-9, -0.5) {$\dots$};
		\node [style=none] (52) at (-9, -2.5) {$\dots$};
		\node [style=none] (53) at (-8.25, 2.5) {};
		\node [style=none] (54) at (-6.5, -2.25) {};
		\node [style=none] (55) at (-8.25, -2.75) {};
		\node [style=none] (56) at (-6.5, -2.75) {};
		\node [style=none] (57) at (-6.5, 2) {};
		\node [style=none] (58) at (-6.5, -0.75) {};
		\node [style=none] (59) at (-8.25, -0.25) {};
		\node [style=none] (60) at (-8.25, -0.75) {};
		\node [style=none] (61) at (-6.5, -0.25) {};
		\node [style=none] (62) at (-6.5, 2.5) {};
		\node [style=none] (63) at (-8.25, -2.25) {};
		\node [style=none] (64) at (-8.25, 2) {};
		\node [style=none] (65) at (5.5, 1) {$\vdots$};
		\node [style=none] (66) at (-15.5, 0) {};
		\node [style=state] (67) at (-17.5, 0) {$b$};
		\node [style=box] (68) at (-12, -0.5) {\small$\big(M_{N-1}^{(R)}\big)^\dagger$};
		\node [style=none] (69) at (-12, 1) {$\vdots$};
		\node [style=none] (70) at (-14.5, 2.25) {};
		\node [style=none] (71) at (-12.5, -2.5) {};
		\node [style=none] (72) at (-14.5, -0.5) {};
		\node [style=box] (73) at (-12, -2.5) {\small$\big(M_{N}^{(R)}\big)^\dagger$};
		\node [style=none] (74) at (-14.5, -2.5) {};
		\node [style=none] (75) at (-12.5, -0.5) {};
		\node [style=box] (76) at (-12, 2.25) {\small$\big(M_{1}^{(R)}\big)^\dagger$};
		\node [style=none] (77) at (-12.5, 2.25) {};
		\node [style=labelnode] (78) at (7.5, 0) {$=$};
		\node [style=stateLargerThin] (79) at (9.5, 0) {$\psi_b$};
		\node [style=labelnode] (80) at (12.5, -2.25) {$\SpaceH_{N}$};
		\node [style=none] (81) at (10, -2.25) {};
		\node [style=none] (82) at (10, -0.25) {};
		\node [style=labelnode] (83) at (12.5, 2.5) {$\SpaceH_1$};
		\node [style=none] (84) at (10, 2.5) {};
		\node [style=labelnode] (85) at (12.5, -0.25) {$\SpaceH_{N-1}$};
		\node [style=none] (86) at (12.5, 1) {$\vdots$};
	\end{pgfonlayer}
	\begin{pgfonlayer}{edgelayer}
		\node [style=boxLargerThin, fill=gray!25] (87) at (-3, 0) {};
		\node [style=boxLargerThin, fill=gray!25] (88) at (-15, 0) {};
		\node [style=effectLargerThin, fill=gray!25] (89) at (3, 0) {};
		\draw [style=dashed] (3.center) to (4.center);
		\draw [style=dashed] (6.center) to (5.center);
		\draw [style=dashed, in=180, out=0] (8.center) to (7.center);
		\draw [style=-, line width=4 pt, draw=white, in=0, out=180] (9) to (10.center);
		\draw [style=-] (9) to (10.center);
		\draw [style=-, line width=4 pt, draw=white, in=0, out=180] (12) to (11.center);
		\draw [style=-] (12) to (11.center);
		\draw [style=-, line width=4 pt, draw=white, in=0, out=180] (14) to (13.center);
		\draw [style=-] (14) to (13.center);
		\draw [style=dashed] (21.center) to (31.center);
		\draw [style=dashed] (20.center) to (27.center);
		\draw [style=dashed, in=180, out=0] (16.center) to (23.center);
		\draw [style=-, line width=4 pt, draw=white, in=180, out=0] (25.center) to (22.center);
		\draw [style=-] (25.center) to (22.center);
		\draw [style=-, line width=4 pt, draw=white, in=180, out=0] (24.center) to (17.center);
		\draw [style=-] (24.center) to (17.center);
		\draw [style=-, line width=4 pt, draw=white, in=180, out=0] (30.center) to (19.center);
		\draw [style=-] (30.center) to (19.center);
		\draw [style=dashed] (33.center) to (32.center);
		\draw [style=dashed] (35.center) to (34.center);
		\draw [style=dashed] (37.center) to (36.center);
		\draw [style=-, line width=4 pt, draw=white, in=180, out=0] (43.center) to (42.center);
		\draw [style=-] (43.center) to (42.center);
		\draw [style=-, line width=4 pt, draw=white, in=180, out=0] (39.center) to (38.center);
		\draw [style=-] (39.center) to (38.center);
		\draw [style=-, line width=4 pt, draw=white, in=180, out=0] (41.center) to (40.center);
		\draw [style=-] (41.center) to (40.center);
		\draw [style=dashed] (45.center) to (44.center);
		\draw [style=dashed] (47.center) to (46.center);
		\draw [style=dashed] (49.center) to (48.center);
		\draw [style=-, line width=4 pt, draw=white, in=180, out=0] (53.center) to (62.center);
		\draw [style=-] (53.center) to (62.center);
		\draw [style=-, line width=4 pt, draw=white, in=180, out=0] (59.center) to (61.center);
		\draw [style=-] (59.center) to (61.center);
		\draw [style=-, line width=4 pt, draw=white, in=180, out=0] (63.center) to (54.center);
		\draw [style=-] (63.center) to (54.center);
		\draw [style=dashed] (64.center) to (57.center);
		\draw [style=dashed] (60.center) to (58.center);
		\draw [style=dashed] (55.center) to (56.center);
		\draw [style=dashed] (66.center) to (67);
		\draw [style=dashed] (70.center) to (77.center);
		\draw [style=dashed] (72.center) to (75.center);
		\draw [style=dashed] (74.center) to (71.center);
		\draw [style=-, line width=4 pt, draw=white, in=0, out=180] (83) to (84.center);
		\draw [style=-] (83) to (84.center);
		\draw [style=-, line width=4 pt, draw=white, in=0, out=180] (85) to (82.center);
		\draw [style=-] (85) to (82.center);
		\draw [style=-, line width=4 pt, draw=white, in=0, out=180] (80) to (81.center);
		\draw [style=-] (80) to (81.center);
	\end{pgfonlayer}
\end{tikzpicture}}
\end{equation}
By repeatedly invoking Equation \ref{LOCCinstrument}, i.e. by inserting classical resolutions of the identity on all classical wires, the LHS of Equation \ref{LOCCprotocolTimeReversedCondition} turns into a mixture of product states. But RHS of Equation \ref{LOCCprotocolTimeReversedCondition} is a pure state, and by the very definition of purity we have that the mixture on the LHS is necessarily trivial. We conclude that $\psi_b$ must be a pure product state, for each choice of $b \in B$: under the assumption that the players can deterministically distinguish all states in the basis, we have shown that the basis cannot contain any entangled states.

Conversely, assume that the family $(\psi_b)_{b \in B}$ contains only product states, and write $\psi_b = \otimes_{i=1}^{N} \psi_{b,i}$. Without loss of generality, assume that all the local states $\psi_{b,i}$ are pure (because $\psi_b$ is) and normalised (because $\psi_b$ is normalised, and hence all $\psi_{b,i}$ must be normalisable). Then the family can be prepared by using a normalised LOCC protocol as follows:
\begin{equation}\label{LOCCprotocolProdPreparation}
	\resizebox{0.9\textwidth}{!}{\begin{tikzpicture}
	\begin{pgfonlayer}{nodelayer}

		\node [style=state] (0) at (-17, 0) {$b$};
		\node [style=boxLargerThin, fill=gray!25] (1) at (-15, 0) {};
		\node [style=none] (2) at (-14.5, 2.25) {};
		\node [style=box] (3) at (-12, 2.25) {\small$\hspace{2.5mm}M_{1}\hspace{2.5mm}$};
		\node [style=none] (4) at (-12, 1) {$\vdots$};
		\node [style=none] (5) at (-14.5, -0.5) {};
		\node [style=box] (6) at (-12, -0.5) {\small$M_{N-1}$};
		\node [style=none] (7) at (-14.5, -2.5) {};
		\node [style=box] (8) at (-12, -2.5) {\small$\hspace{2mm}M_{N}\hspace{2mm}$};
		\node [style=labelnode] (9) at (-9, 2.25) {$\SpaceH_1$};
		\node [style=labelnode] (10) at (-9, -0.5) {$\SpaceH_{N-1}$};
		\node [style=labelnode] (11) at (-9, -2.5) {$\SpaceH_{N}$};
		\node [style=none] (12) at (-7,0) {$=$};

		\node [style=state] (13) at (-5, 2.25) {$b$};
		\node [style=box] (14) at (-2.5, 2.25) {\small$\hspace{2.5mm}M_{1}\hspace{2.5mm}$};
		\node [style=none] (15) at (-2.5, 1) {$\vdots$};
		\node [style=state] (16) at (-5, -0.5) {$b$};
		\node [style=box] (17) at (-2.5, -0.5) {\small$M_{N-1}$};
		\node [style=state] (18) at (-5, -2.5) {$b$};
		\node [style=box] (19) at (-2.5, -2.5) {\small$\hspace{2mm}M_{N}\hspace{2mm}$};
		\node [style=labelnode] (20) at (0.5, 2.25) {$\SpaceH_1$};
		\node [style=labelnode] (21) at (0.5, -0.5) {$\SpaceH_{N-1}$};
		\node [style=labelnode] (22) at (0.5, -2.5) {$\SpaceH_{N}$};
		\node [style=none] (23) at (2.5,0) {$=$};

		\node [style=state] (24) at (5, 2.25) {\small$\hspace{2.5mm}\psi_{b,1}\hspace{2.5mm}$};
		\node [style=none] (25) at (5, 1) {$\vdots$};
		\node [style=state] (26) at (5, -0.5) {\small$\psi_{b,N-1}$};
		\node [style=state] (27) at (5, -2.5) {\small$\hspace{2mm}\psi_{b,N}\hspace{2mm}$};
		\node [style=labelnode] (28) at (8, 2.25) {$\SpaceH_1$};
		\node [style=labelnode] (29) at (8, -0.5) {$\SpaceH_{N-1}$};
		\node [style=labelnode] (30) at (8, -2.5) {$\SpaceH_{N}$};
		\node [style=none] (31) at (10,0) {$=$};

		\node [style=stateLargerThin] (32) at (12, 0) {$\psi_b$};
		\node [style=none] (33) at (12, 2.25) {};
		\node [style=labelnode] (34) at (15, 2.25) {$\SpaceH_1$};		
		\node [style=none] (35) at (12, -0.5) {};
		\node [style=labelnode] (36) at (15, -0.5) {$\SpaceH_{N-1}$};
		\node [style=none] (37) at (12, -2.5) {};
		\node [style=labelnode] (38) at (15, -2.5) {$\SpaceH_{N}$};
		\node [style=none] (39) at (15, 1) {$\vdots$};

	\end{pgfonlayer}
	\begin{pgfonlayer}{edgelayer}
		\draw [style=dashed] (0) to (1);

		\draw [style=dashed] (2) to (3);
		\draw [style=dashed] (5) to (6);
		\draw [style=dashed] (7) to (8);
		\draw [style=-] (3) to (9);
		\draw [style=-] (6) to (10);
		\draw [style=-] (8) to (11);

		\draw [style=dashed] (13) to (14);
		\draw [style=dashed] (16) to (17);
		\draw [style=dashed] (18) to (19);
		\draw [style=-] (14) to (20);
		\draw [style=-] (17) to (21);
		\draw [style=-] (19) to (22);

		\draw [style=-] (24) to (28);
		\draw [style=-] (26) to (29);
		\draw [style=-] (27) to (30);

		\draw [style=-] (33) to (34);
		\draw [style=-] (35) to (36);
		\draw [style=-] (37) to (38);
	\end{pgfonlayer}
\end{tikzpicture}}
\end{equation}
The time-reversal of the LOCC protocol described by Diagram \ref{LOCCprotocolProdPreparation} is a unital LOCC protocol which implements the state distinguishing task for the family $(\psi_b)_{b \in B}$. The time-reversed LOCC protocol is normalised (resp. unital) if and only if the original protocol described by Diagram \ref{LOCCprotocolProdPreparation} is unital (resp. normalised).

Finally, we can put these results together. If the players can perfectly distinguish between the states of a complete orthonormal family by a normalised unital LOCC protocol, then the family can be prepared by the time-reversed LOCC protocol, which is unital and normalised, and hence it cannot contain any entangled states. Vice versa, if the states in a complete orthonormal family can be perfectly prepared by a normalised unital LOCC protocol, then the family does not contain any product states and it can furthermore be distinguished by the time-reversed LOCC protocol, which is unital and normalised.
\end{proof}

\section{Discussion}

We have provided a simple argument showing that any complete orthonormal family of multipartite pure states which can be perfectly distinguished by LOCC protocols cannot contain any entangled states. Our proof is diagrammatic and theory-independent, and straightforwardly applies to both quantum theory and any post-quantum theory which can be modelled by our categorical framework.
A number of well-established results in LOCC-distinguishability arise as a corollary of our work: for example, the fact that the four two-qubit Bell states are not LOCC-distinguishable \cite{Getal01}, or the fact that the four two-qubit states $\{\ket{00}, \ket{11}, \ket{01}+\ket{10},\ket{01}-\ket{10}\}$ are not LOCC-distinguishable \cite{Getal02}. Contrary to the majority of previous results, our treatment is independent of the number of parties and of the dimension of quantum systems involved.

Our proof shows that LOCC-distinguishability of complete orthonormal families is really about purity and time-reversal, bearing no relation to normalisation and causality. In this sense, it is a story about the shape of processes, rather than their inner workings. The result also finds a particularly fitting interpretation in the resource theory of entanglement: free processes can distinguish between states of a complete orthonormal family only when the states themselves are all free.

Future work will focus on extending our theory-independent diagrammatic approach to more general problems of interest in LOCC-distinguishability and the resource theory of entanglement. For example, we propose to investigate the complete orthonormal family of \cite{Betal99} from a process shape perspective, hopefully shedding further light on an otherwise counter-intuitive result.

\begin{acknowledgments}
	SG gratefully acknowledges funding from \mbox{EPSRC} and the \mbox{Trinity} \mbox{College} \mbox{Williams} Scholarship. SRM gratefully acknowledges funding from the \mbox{Clarendon Fund}, the \mbox{Keble} \mbox{College} \mbox{Sloane} \mbox{Robinson} Award, and the \mbox{Oxford-DeepMind} Graduate Scholarship. This publication was made possible through the support of a grant from the John Templeton Foundation. The opinions expressed in this publication are those of the authors and do not necessarily reflect the views of the John Templeton Foundation.
\end{acknowledgments}


\begin{thebibliography}{41}
\providecommand{\natexlab}[1]{#1}
\providecommand{\url}[1]{\texttt{#1}}
\expandafter\ifx\csname urlstyle\endcsname\relax
  \providecommand{\doi}[1]{doi: #1}\else
  \providecommand{\doi}{doi: \begingroup \urlstyle{rm}\Url}\fi

\bibitem[Abramsky and Coecke(2004)]{Abramsky2004}
Samson Abramsky and Bob Coecke.
\newblock {A categorical semantics of quantum protocols}.
\newblock In \emph{Proceedings of the 19th Annual IEEE Symposium on Logic in
  Computer Science, 2004.}, pages 415--425. IEEE, 2004.
\newblock \doi{10.1109/LICS.2004.1319636}.

\bibitem[Abramsky and Coecke(2009)]{Abramsky2009}
Samson Abramsky and Bob Coecke.
\newblock {Categorical Quantum Mechanics}.
\newblock \emph{Handbook of Quantum Logic and Quantum Structures}, pages
  261--323, 2009.
\newblock \doi{10.1016/B978-0-444-52869-8.50010-4}.

\bibitem[Backens(2014)]{Backens2014}
Miriam Backens.
\newblock {The ZX-calculus is complete for stabilizer quantum mechanics}.
\newblock \emph{New Journal of Physics}, 16\penalty0 (9), 2014.
\newblock \doi{10.1088/1367-2630/16/9/093021}.

\bibitem[Bennett et~al.(1999)Bennett, DiVincenzo, Fuchs, Mor, Rains, Shor,
  Smolin, and Wootters]{Betal99}
Charles~H Bennett, David~P DiVincenzo, Christopher~A Fuchs, Tal Mor, Eric
  Rains, Peter~W Shor, John~A Smolin, and William~K Wootters.
\newblock Quantum nonlocality without entanglement.
\newblock \emph{Physical Review A}, 59\penalty0 (2):\penalty0 1070, 1999.

\bibitem[Chefles(2004)]{C04}
Anthony Chefles.
\newblock Condition for unambiguous state discrimination using local operations
  and classical communication.
\newblock \emph{Physical Review A}, 69\penalty0 (5):\penalty0 050307, 2004.

\bibitem[Childs et~al.(2013)Childs, Leung, Man{\v{c}}inska, and Ozols]{CLMO13}
Andrew~M Childs, Debbie Leung, Laura Man{\v{c}}inska, and Maris Ozols.
\newblock A framework for bounding nonlocality of state discrimination.
\newblock \emph{Communications in Mathematical Physics}, 323\penalty0
  (3):\penalty0 1121--1153, 2013.
\newblock \doi{10.1007/s00220-013-1784-0}.

\bibitem[Chiribella(2014)]{Chiribella2014}
Giulio Chiribella.
\newblock {Dilation of states and processes in operational-probabilistic
  theories}.
\newblock \emph{EPTCS 172}, 2014.
\newblock \doi{10.4204/EPTCS.172.1}.

\bibitem[Chiribella and Scandolo(2015)]{Chiribella2015}
Giulio Chiribella and Carlo~Maria Scandolo.
\newblock {Entanglement and thermodynamics in general probabilistic theories}.
\newblock \emph{New Journal of Physics}, 17\penalty0 (10):\penalty0 103027,
  2015.
\newblock \doi{10.1088/1367-2630/17/10/103027}.

\bibitem[Chiribella et~al.(2010)Chiribella, D'Ariano, and
  Perinotti]{Chiribella2010}
Giulio Chiribella, Giacomo~Mauro D'Ariano, and Paolo Perinotti.
\newblock {Probabilistic theories with purification}.
\newblock \emph{Physical Review A - Atomic, Molecular, and Optical Physics},
  81\penalty0 (6), 2010.
\newblock \doi{10.1103/PhysRevA.81.062348}.

\bibitem[Chiribella et~al.(2011)Chiribella, D'Ariano, and
  Perinotti]{Chiribella2011}
Giulio Chiribella, Giacomo~Mauro D'Ariano, and Paolo Perinotti.
\newblock {Informational derivation of quantum theory}.
\newblock \emph{Physical Review A - Atomic, Molecular, and Optical Physics},
  84\penalty0 (1):\penalty0 1--39, 2011.
\newblock \doi{10.1103/PhysRevA.84.012311}.

\bibitem[Chiribella et~al.(2016)Chiribella, D'Ariano, and
  Perinotti]{Chiribella2016}
Giulio Chiribella, Giacomo~Mauro D'Ariano, and Paolo Perinotti.
\newblock {Quantum from principles}.
\newblock In Giulio Chiribella and Robert~W Spekkens, editors, \emph{{Quantum
  Theory: Informational Foundations and Foils}}. {Springer}, 2016.
\newblock \doi{10.1007/978-94-017-7303-4}.

\bibitem[Chitambar et~al.(2014{\natexlab{a}})Chitambar, Duan, and Hsieh]{CDH14}
Eric Chitambar, Runyao Duan, and Min-Hsiu Hsieh.
\newblock When do local operations and classical communication suffice for
  two-qubit state discrimination?
\newblock \emph{IEEE Transactions on Information Theory}, 60\penalty0
  (3):\penalty0 1549--1561, 2014{\natexlab{a}}.
\newblock \doi{10.1109/TIT.2013.2295356}.

\bibitem[Chitambar et~al.(2014{\natexlab{b}})Chitambar, Leung, Man{\v{c}}inska,
  Ozols, and Winter]{Cetal14}
Eric Chitambar, Debbie Leung, Laura Man{\v{c}}inska, Maris Ozols, and Andreas
  Winter.
\newblock Everything you always wanted to know about locc (but were afraid to
  ask).
\newblock \emph{Communications in Mathematical Physics}, 328\penalty0
  (1):\penalty0 303--326, 2014{\natexlab{b}}.
\newblock \doi{10.1007/s00220-014-1953-9}.

\bibitem[Coecke(2009)]{Coecke2009picturalism}
Bob Coecke.
\newblock {Quantum Picturalism}.
\newblock \emph{Contemporary Physics}, pages 1--32, 2009.

\bibitem[Coecke and Duncan(2011)]{Coecke2011}
Bob Coecke and Ross Duncan.
\newblock {Interacting quantum observables: Categorical algebra and
  diagrammatics}.
\newblock \emph{New Journal of Physics}, 13, 2011.
\newblock \doi{10.1088/1367-2630/13/4/043016}.

\bibitem[Coecke and Kissinger(2015)]{Coecke2015}
Bob Coecke and Aleks Kissinger.
\newblock {Categorical Quantum Mechanics I: Causal Quantum Processes}.
\newblock 2015.
\newblock URL \url{http://arxiv.org/abs/1510.05468}.

\bibitem[Coecke and Kissinger(2016)]{Coecke2016}
Bob Coecke and Aleks Kissinger.
\newblock {Categorical Quantum Mechanics II: Classical-Quantum Interaction}.
\newblock 2016.
\newblock URL \url{http://arxiv.org/abs/1605.08617}.

\bibitem[Coecke and Kissinger(2017)]{Coecke2017}
Bob Coecke and Aleks Kissinger.
\newblock \emph{{Picturing Quantum Processes}}.
\newblock Cambridge University Press, 2017.

\bibitem[Coecke and Lal(2013)]{Coecke2013a}
Bob Coecke and Raymond Lal.
\newblock {Causal Categories: Relativistically Interacting Processes}.
\newblock \emph{Foundations of Physics}, 43\penalty0 (4):\penalty0 458--501,
  2013.
\newblock \doi{10.1007/s10701-012-9646-8}.

\bibitem[Coecke et~al.(2016)Coecke, Fritz, and Spekkens]{Coecke2016b}
Bob Coecke, Tobias Fritz, and Robert~W. Spekkens.
\newblock {A mathematical theory of resources}.
\newblock \emph{Information and Computation}, 250:\penalty0 59--86, 2016.
\newblock \doi{10.1016/j.ic.2016.02.008}.

\bibitem[D'Ariano et~al.(2017)D'Ariano, Chiribella, and Perinotti]{DAriano2017}
Giacomo~Mauro D'Ariano, Giulio Chiribella, and Paolo Perinotti.
\newblock \emph{{Quantum theory from first principles}}.
\newblock Cambridge University Press, 2017.

\bibitem[de~Beaudrap and Horsman(2017)]{Horsman2017}
Niel de~Beaudrap and Dominic~C Horsman.
\newblock {The ZX calculus is a language for surface code lattice surgery}.
\newblock 2017.
\newblock URL \url{http://arxiv.org/abs/1704.08670}.

\bibitem[DiVincenzo et~al.(2003)DiVincenzo, Hayden, and Terhal]{DHT03}
David~P DiVincenzo, Patrick Hayden, and Barbara~M Terhal.
\newblock Hiding quantum data.
\newblock \emph{Foundations of Physics}, 33\penalty0 (11):\penalty0 1629--1647,
  2003.

\bibitem[Duan et~al.(2009)Duan, Feng, Xin, and Ying]{DFXY09}
Runyao Duan, Yuan Feng, Yu~Xin, and Mingsheng Ying.
\newblock Distinguishability of quantum states by separable operations.
\newblock \emph{IEEE Transactions on Information Theory}, 55\penalty0
  (3):\penalty0 1320--1330, 2009.

\bibitem[Feng and Shi(2009)]{FS09}
Yuan Feng and Yaoyun Shi.
\newblock Characterizing locally indistinguishable orthogonal product states.
\newblock \emph{IEEE Transactions on Information Theory}, 55\penalty0
  (6):\penalty0 2799--2806, 2009.

\bibitem[Ghosh et~al.(2001)Ghosh, Kar, Roy, Sen, Sen, et~al.]{Getal01}
Sibasish Ghosh, Guruprasad Kar, Anirban Roy, Aditi Sen, Ujjwal Sen, et~al.
\newblock Distinguishability of bell states.
\newblock \emph{Physical review letters}, 87\penalty0 (27):\penalty0 277902,
  2001.

\bibitem[Ghosh et~al.(2002)Ghosh, Kar, Roy, Sarkar, Sen, Sen, et~al.]{Getal02}
Sibasish Ghosh, Guruprasad Kar, Anirban Roy, Debasis Sarkar, Aditi Sen, Ujjwal
  Sen, et~al.
\newblock Local indistinguishability of orthogonal pure states by using a bound
  on distillable entanglement.
\newblock \emph{Physical Review A}, 65\penalty0 (6):\penalty0 062307, 2002.

\bibitem[Gogioso and Scandolo(2017)]{GS17}
Stefano Gogioso and Carlo~Maria Scandolo.
\newblock Categorical probabilistic theories.
\newblock \emph{EPSRC (QPL 2017)}, 2017.
\newblock URL \url{http://arxiv.org/abs/1701.08075}.

\bibitem[Hardy(2001)]{Hardy2001}
Lucien Hardy.
\newblock {Quantum Theory From Five Reasonable Axioms}.
\newblock 2001.
\newblock URL \url{http://arxiv.org/abs/quant-ph/0101012}.

\bibitem[Hardy(2011{\natexlab{a}})]{Hardy2011}
Lucien Hardy.
\newblock {Foliable operational structures for general probabilistic theories}.
\newblock In Hans Halvorson, editor, \emph{{Deep Beauty}}. {Cambridge
  University Press}, 2011{\natexlab{a}}.
\newblock \doi{10.1017/CBO9780511976971.013}.

\bibitem[Hardy(2011{\natexlab{b}})]{Hardy2011informational}
Lucien Hardy.
\newblock {Reformulating and reconstructing quantum theory}.
\newblock 2011{\natexlab{b}}.
\newblock URL \url{http://arxiv.org/abs/1104.2066}.

\bibitem[Hardy(2016)]{Hardy2013}
Lucien Hardy.
\newblock {Reconstructing Quantum Theory}.
\newblock In Giulio Chiribella and Robert~W Spekkens, editors, \emph{{Quantum
  Theory: Informational Foundations and Foils}}. {Springer}, 2016.
\newblock \doi{10.1007/978-94-017-7303-4\_7}.

\bibitem[Hillery and Mimih(2003)]{HM03}
Mark Hillery and Jihane Mimih.
\newblock Distinguishing two-qubit states using local measurements and
  restricted classical communication.
\newblock \emph{Physical Review A}, 67\penalty0 (4):\penalty0 042304, 2003.

\bibitem[Horodecki et~al.(2003)Horodecki, Sen, Sen, and Horodecki]{HSSK03}
Micha{\l} Horodecki, Aditi Sen, Ujjwal Sen, and Karol Horodecki.
\newblock Local indistinguishability: More nonlocality with less entanglement.
\newblock \emph{Physical review letters}, 90\penalty0 (4):\penalty0 047902,
  2003.

\bibitem[Horsman(2011)]{Horsman2011}
Clare Horsman.
\newblock {Quantum picturalism for topological cluster-state}.
\newblock \emph{New Journal of Physics}, 133\penalty0 (9), 2011.
\newblock \doi{10.1088/1367-2630/13/9/095011}.

\bibitem[Lancien and Winter(2012)]{LW13}
C{\'e}cilia Lancien and Andreas Winter.
\newblock Distinguishing multi-partite states by local measurements.
\newblock \emph{Communications in Mathematical Physics}, 323\penalty0
  (2):\penalty0 555--573, 2012.

\bibitem[Peres(2004)]{P04}
Asher Peres.
\newblock What is actually teleported?
\newblock \emph{IBM Journal of Research and Development}, 48\penalty0
  (1):\penalty0 63--69, 2004.

\bibitem[Roy~Moulik and Panigrahi(2016)]{RoyMoulik2016}
Subhayan Roy~Moulik and Prasanta~K. Panigrahi.
\newblock {Timelike curves can increase entanglement with LOCC}.
\newblock \emph{Scientific Reports}, 6, 2016.
\newblock \doi{10.1038/srep37958}.

\bibitem[Walgate and Hardy(2002)]{WH02}
Jonathan Walgate and Lucien Hardy.
\newblock Nonlocality, asymmetry, and distinguishing bipartite states.
\newblock \emph{Physical review letters}, 89\penalty0 (14):\penalty0 147901,
  2002.

\bibitem[Walgate et~al.(2000)Walgate, Short, Hardy, and Vedral]{Wetal00}
Jonathan Walgate, Anthony~J Short, Lucien Hardy, and Vlatko Vedral.
\newblock Local distinguishability of multipartite orthogonal quantum states.
\newblock \emph{Physical Review Letters}, 85\penalty0 (23):\penalty0 4972,
  2000.

\bibitem[Zhang et~al.(2016)Zhang, Tan, Weng, and Li]{ZTWL16}
Xiaoqian Zhang, Xiaoqing Tan, Jian Weng, and Yongjun Li.
\newblock Locc indistinguishable orthogonal product quantum states.
\newblock \emph{Scientific reports}, 6, 2016.
\newblock \doi{doi:10.1038/srep28864}.

\end{thebibliography}

\end{document}